%% file: main.tex
\documentclass[letterpaper,11pt]{article}
\usepackage{fullpage}
\usepackage[hmargin=1in,vmargin=1in]{geometry}
\usepackage[english]{babel}
\usepackage[utf8]{inputenc}
\usepackage{amsmath,amsthm}
\usepackage{algorithmic}
\usepackage{algorithm}
\usepackage{enumerate}  
\usepackage{subfig}
\usepackage{amssymb}
\usepackage{lipsum}
\usepackage{mathtools}
\usepackage{graphicx,color}
\usepackage{authblk}
\usepackage{hyperref}
\usepackage{booktabs} 
\usepackage{tikz,subfig}
\usetikzlibrary{matrix,arrows,decorations.pathreplacing,automata,positioning,decorations.pathmorphing}
\usepackage{float}

\newtheorem{lemma}{Lemma}
\newtheorem{theorem}{Theorem}
\newtheorem{definition}{Definition}

\title{Random Walk-based In-network Computation of Arbitrary Functions}

\author[1]{Iqra Altaf Gillani}
\author[2]{Pooja Vyavahare}
\author[1]{Amitabha Bagchi}
\affil[1]{Department of Computer Science and Engineering, IIT Delhi\\ \{iqraaltaf,bagchi\}@cse.iitd.ac.in}
\affil[2]{School of Computing and Electrical Engineering, IIT Mandi\\ pooja\_vyavahare@iitmandi.ac.in}

\newcommand{\nbd}[1]{\mathsf{Nbd}(#1)}
\newcommand{\degree}[1]{\mbox{deg}(#1)}
\newcommand{\prob}[1]{\mathbb{P}\left[ #1 \right]}

\newcommand{\comm}[1]{}
\newcommand{\tmix}{t_{\mbox{\scriptsize mix}}}
\newcommand{\tbarfk}{\bar{\tau}_{\mbox{\scriptsize f(K)}}}

\newcommand{\probm}[1]{\mathcal{P}[ #1 ]}
\newcommand{\thits}{t_{\mbox{\scriptsize hit}}}

\newcommand{\tcomp}{\tau_{\mbox{\scriptsize f(K)}}^\ell}
\newcommand{\tcompk}[1]{\tau_{\mbox{\scriptsize comp}}(#1)}
\newcommand{\tappk}{\tau_{\mbox{\scriptsize app}}^\ell}
\newcommand{\sink}{u_{\mbox{\scriptsize s}}}
\newcommand{\T}{\mathcal{T}}

\newcommand{\fixedalgo}{Fixed Random-Compute}
\newcommand{\flexiblealgo}{Flexible Random-Compute}

\newcommand{\ex}[1]{{\mbox{E}\left[#1\right]}}

\providecommand{\keywords}[1]{\textbf{\textit{Keywords:}}#1}
\date{}

\begin{document}
\maketitle

\begin{abstract}
We study in-network computation on general network topologies. Specifically, we are given the description of a function, and a network with distinct nodes at which the operands of the function are made available, and a designated sink where the computed value of the function is to be consumed. We want to compute the function during the process of moving the data towards the sink. Such settings have been studied in the literature, but mainly for symmetric functions, e.g. average, parity etc., which have the specific property that the output is invariant to permutation of the operands. To the best of our knowledge, we present the first fully decentralised algorithms for arbitrary functions, which we model as those functions whose computation schema is structured as a binary tree. We propose two algorithms, Fixed Random-Compute and Flexible Random-Compute, for this problem, both of which use simple random walks on the network as their basic primitive. Assuming a stochastic model for the generation of streams of data at each source, we provide a lower and an upper bound on the rate at which Fixed Random-Compute can compute the stream of associated function values. Note that the lower bound on rate though computed for our algorithm serves as a general lower bound for the function computation problem and to the best of our knowledge is first such lower bound for asymmetric functions. We also provide upper bounds on the average time taken to compute the function, characterising this time in terms of the fundamental parameters of the random walk on the network: the hitting time in the case of Fixed Random-Compute, and the mixing time in the case of Flexible Random-Compute.
 \end{abstract}
 

 \section{Introduction}
\label{sec:intro}
Since most commercially available sensor nodes used in today's sensor networks are capable of performing small operations on the data, distributed function computation (also known as in-network computation) algorithms seek to exploit the computation capability of these nodes to increase the efficiency of communication and computation of function~\cite{Kannan-JSAC:2013,Khude-INFOCOM:2005,Kamath-TIT:2014} over plain data forwarding techniques~\cite{Kamra-SIGCOMM:2006,Intanagonwiwat-TN:2003}. The in-network computation paradigm is based on the following simple idea: instead of moving all the distinct data items generated at different nodes of the network to the sink and computing the function of interest at the sink, we leverage the meetings of the data items at intermediate nodes to compute partial functions which can then be combined at the sink. The idea is to reduce the load on the network and thereby increase the rate at which the data can be read.

 Clearly, the key question here is: is it always possible to combine any two data items that meet at a node? In the simplest scenario, the answer is yes, and the class of functions for which the answer is yes are known in the literature as ``symmetric functions'' and most of the existing literature on in-network function computation mainly deals with such functions (see, e.g.,\cite{Giridhar-JSAC:2005}).
In this paper, our aim is to study the distributed computation of arbitrary functions where the sequence of operations of the data is important to compute the final function value. Generalising from the class of symmetric functions requires us to be able to specify the kinds of functions we handle and typically functions are specified by a {\em computation schema} which describes the structure of allowable combinations of operands and partial functions that make up the final function being computed. In particular, we study a class of functions whose \emph{computation schema} can be represented by a binary tree. Another major contribution of our work is to distinguish between two modes of computing intermediate values (partial functions) that make up a larger function contribution, the distinction made on the basis of whether the node at which a specific intermediate value is computed is specified in advance (we call it the ``Fixed'' scenario) or not (the ``Flexible'' scenario.)

The distributed computation of functions with computation schema modelled by a directed graph has recently attracted some attention in the literature \cite{Shah-JSAC:2013,Liu-SIGMETRICS:2013,Vyavahare-arXiv:2015}, but these works typically take a centralised approach to the problem, centrally computing routings that realise the function within the network. Our approach, based on random walks, is fundamentally decentralised. Such a communication strategy is very useful when the network is changing constantly and the routing information can become invalid frequently. Any node of the network in our algorithm does not need global information (like the number of nodes or topology of the network) and the communication depends only on the knowledge of the neighbourhood. In our algorithms, multiple data packets move across the network leading to multiple random walks in the network. In order to compute the function in the network, our algorithms combine these data packets in the order defined by the computation schema of the function. This combination of packets leads to \emph{coalescence} of random walks which have been studied in \cite{Cooper-SIAM:2013,Kanade-arXiv:2017}.

\paragraph*{Our contributions} 
\begin{enumerate}
\item We describe two in-network computation scenarios: (a) the {\em fixed} scenario under which each internal node of the function schema is mapped to a specific network node that is tasked with computing the subfunction corresponding to that internal node of the schema and (b) the {\em flexible} scenario in which a subfunction computation takes place opportunistically at any network node that happens to have the two relevant operands.
\item We propose simple, decentralised random walk-based
algorithms, \fixedalgo\ and \flexiblealgo, to compute a function with a binary tree
computation schema in the fixed and flexible network scenarios respectively. 
\item We present upper and lower bounds on the rate of computation for \fixedalgo\ in a setting where data is generated stochastically in rounds. In particular, the lower bound presented is a general lower bound for the function computation problem and to the best of our knowledge is the first of its kind for asymmetric functions. We also find the average function computation time taken by these algorithms.
\end{enumerate}

\paragraph*{Related work}

As discussed earlier, most of the extant literature study the in-network computation of symmetric functions (see \cite{Giridhar-JSAC:2005,Mosk-PODC:2006,Dutta-FOCS:2008,Iyer-TMC:2011}) where the sequence of computation does not matter. Researchers have studied the computation time as well as the rate of computation of various classes of symmetric functions for different network models. In particular, \cite{Giridhar-JSAC:2005} study the rate of computation of classes of symmetric functions namely \emph{type-thereshold} (example: maximum) and \emph{type-sensitive} (example: average) for collocated and random planar multihop networks. Following the work of \cite{Giridhar-JSAC:2005}, various authors studied these functions in various network settings; for example \cite{Dutta-FOCS:2008} studied for noisy communication model, \cite{Karamchandani-TIT:2011} for grid networks and \cite{Khude-INFOCOM:2005} for random geometric graphs. Further, Banerjee et. al. \cite{Banerjee-QS:2012} study the rate of computation of \emph{divisible} functions which can be computed by a divide-and-conquer method on any subset of the source data. Another variant studied in the literature is of \emph{$\lambda_f$-divisible} functions which can be computed by performing operations on at most $\lambda_f$ source data values at any time~\cite{Kannan-JSAC:2013}. Note that all the work discussed above only present upper bounds on the rate. However, Kamath et al. \cite{Kamath-TIT:2014} find a lower bound for computing a symmetric function \texttt{MAX} in a pipelined setting different than ours, which requires knowledge about network structure. To the best of our knowledge no lower bound result has been presented for asymmetric functions so far.

In this work, we not only extend the function computation study to the class of arbitrary functions but also present a general lower bound on the rate of function computation which is a first lower bound presented for this class of problem. In particular, we present random walk-based algorithms to compute arbitrary functions with binary tree schema. The communication strategy in our setting is similar to that of \emph{gossip algorithms} \cite{Boyd-INFOCOM:2005,Mosk-PODC:2006,Shah-ICASSP:2009} where each node sends a data packet to a randomly chosen neighbour in every time slot. In this setting, the time to compute the average of data values is studied by \cite{Boyd-INFOCOM:2005} and separable functions of data are studied in \cite{Mosk-PODC:2006}. Separable functions are a class of functions which can be written as the sum of functions on the subsets of source data values. Mosk-Aoyama and Shah  \cite{Mosk-PODC:2006} give an upper bound on the computation of such functions which is directly proportional to the logarithm of the size of the network and inversely proportional to the conductance of the network. The analysis and thereby the results of gossip algorithms depends on the random walk properties which in turn depends on the transition matrix of the random walk. Our results also depend on the random walk parameters like the spectral gap of transition matrix, hitting time and mixing time.

Moreover, in our setting movement of different data operands in the network essentially leads to multiple random walks on the network. These random walks combine with each other based on the sequence defined by the computation schema resulting in the coalescence of random walks. Properties of multiple random walks have been studied in the literature \cite{Alon-CPC:2011,Cooper-ICALP:2009,Efremenko-APPROX-RANDOM:2009,Patel-SIAMX:2016}. Interaction of multiple random walks has been studied by \cite{Cooper-ICALP:2009} where two random walks either annihilate each other or combine when they meet and a bound on the cover time of such walks has been presented. Coalescing random walks (two random walks combine when they meet) has been studied by \cite{Cooper-SIAM:2013,Kanade-arXiv:2017}. An upper bound on the time to coalesce $n$ random walks each starting from different vertices in the graph is given by \cite{Cooper-SIAM:2013}. Recently \cite{Kanade-arXiv:2017} improved this bound on coalescence time for various kinds of graphs. Our analysis of algorithms also involves multiple and coalescing random walks.

\paragraph*{\textbf{Paper organization}} In Section~\ref{sec:prelim}, we describe our in-network computation model, the data generation model (Section~\ref{subsec:ntw_comp_model}), the class of functions we deal with (Section~\ref{subsec:func_schem}) and the two in-network computation scenarios (Section~\ref{subsec:schemes}) that we propose: the fixed and the flexible. In Section~\ref{sec:algos}, we first define the routing scheme used by our algorithms and then we discuss our two random walk-based algorithms for the two scenarios in Section~\ref{subsec:fixed_algo} and Section~\ref{subsec:flexible_algo} respectively. In Section~\ref{sec:metrics_results}, we define our performance metrics and discuss our main results with some examples. In Section~\ref{sec:proofs}, we present the proofs of our theorems. In particular, we discuss the proof of the rate analysis of \fixedalgo\ algorithm in Section~\ref{subsec:rate} and in Section~\ref{subsec:computation_time} we analyse the average time taken by our algorithms in computing an asymmetric function in both fixed and flexible models. We conclude in Section~\ref{sec:conclusion} and provide some directions for future work.

\section{Modelling Assumptions and Computation Scenarios}
\label{sec:prelim}
In this section, we first describe our network, data generation and computation model. Then, we define the particular function schema we work with and discuss the two in-network computing scenarios, the ``Fixed'' and the ``Flexible'' where such functions can be computed.

\subsection{Network and Computation Model}
\label{subsec:ntw_comp_model}
\paragraph{The network model}
The communication network is denoted by an undirected connected graph
$G=(V, E),$ where $V$ is the set of $n$ nodes, $V_s \subseteq V$ with $|V_s|=K$ is the set of source nodes and $E$ is the set of
$m$ edges. An edge $e = (u, v)$ is present between nodes $u,v \in V$
if they can communicate with each other and we denote $\nbd{u}:= \{ v
\in V |  (u, v)\in E \}$. The nodes in the network follow a slotted
time model for communication and we assume that each node can send at
most one data packet to a single neighbour in any time slot. This is
known as the {\em transmitted gossip constraint} in the gossip
literature \cite{Mosk-PODC:2006}. Note that under
this network model a node may receive multiple packets in one time slot.

\paragraph{The model of computation in the network}
We need to compute a function $f_K$ defined as $f_K := f_K(x_1,x_2, \ldots,x_K)$ in the given network. The operand $x_i$ is generated by a \emph{data generation model} at source node $u_i\in V_s$ and we have the source set $V_s=\{u_1, u_s,\ldots, u_K\}$. We need to compute $f_K(x_1,x_2, \ldots,x_K)$ and make it available at a designated node of the network called the sink $\sink\in V$. We work within the paradigm of in-network computation \cite{Banerjee-QS:2012} for computing $f_K$, i.e., nodes can compute intermediate functions called subfunctions of the data while relaying it in the network. We elaborate this further when we discuss the two specific in-network computation scenarios we work with.

\paragraph{The data generation model}
\label{subsec:dat_genration_sink}
We consider the data generation process at each source node
$u_i \in V_s$ as a stochastic arrival process in discrete time that is
Bernoulli with parameter $\beta$ and independent of the arrivals
taking place at all other nodes, i.e. at each time slot $t$ each node
generates a new data packet with probability $\beta$ independent of
all other nodes. We partition the data generated into rounds: the $\ell^{th}$ data
packet to be generated at source node $u_i \in V_s$ is said to be part
of round $\ell$. If we denote this data item as $x_i(l)$ then note that
the function computed in round $\ell$ is $f_K(x_1(\ell), x_2(\ell),
\ldots, x_K(\ell))$. So, each data packet is identified by an identifier $\ell$ which denotes their associated round. We will refer to this model of data generation as
      {\em independent Bernoulli data generation with parameter
        $\beta$}. We will use this data generation model for analysing our algorithms and then in Section~\ref{sec:conclusion} we will also discuss two other data generation models and present our latency results in their context.
        
\subsection{The Binary Tree Function Schema} 
\label{subsec:func_schem}
After defining our network model, model of computation and data generation model now we will define the type of functions we work with i.e. particular class of asymmetric or arbitrary functions. Recall, a function is symmetric if the output value does not depend on the permutation of input values and sequence of intermediate operations does not matter. To generalise from the notion of symmetric to asymmetric we need some language to describe function classes that allow only specific argument combinations. We use the abstraction of a directed acyclic graph to capture the partial combinations that are possible in a given function. We call this graph to be the \emph{computation schema} of the function. In this work, we study a class of functions whose computation schema is a directed tree. An example is shown in Figure~\ref{fig:func_schema}. 

Our algorithms and results are presented for a binary tree computation schema, but this does not, in fact, lead to any great loss in generality when it comes to the class of asymmetric functions. In the function computation schema, any intermediate node representing unary operation can be merged with its parent node as it requires only one operand. The unary operation can be performed by the network node which performs the operation for its parent node. Thus a function computation schema can have only binary and $M$-ary operations. Our results of binary tree schema can easily be extended to arbitrary directed trees (like $M$-ary trees) and we discuss these extensions in Section~\ref{sec:conclusion}. With this in mind, we now formalise the notion of binary tree computation schema and present the notation we will use in the subsequent sections.

Let $f_K := f_K(x_1,x_2, \ldots,x_K)$ be a function whose schema is described by a binary tree $\T$ with $L\leq K$ leaves as sources and height $h$ where $\log K \leq h <K$. For complete binary tree $L=K$, where $K=2^r$ for some value $r$ and $h=\log K$. So, let the root of the tree where finally $f_K$ is computed be at level 0 and have id $\T(0,0)$; note it is the only node at this level. In general, level $i$ has at most $2^i$ nodes with id set $\{\T(i,0), \T(i,1),\ldots,\T(i, 2^i-1)\}$. Let $\theta(i,j)$ be the value computed at node with id $\T(i,j)$. For level $i=h$ i.e. the leaf level, each node performs an identity function on data so, $\theta(i,j)=x_{j+1}$ for $0 \leq j < 2^i$ and in general for level $0 \leq i < h$ and $0 \leq j < 2^i$, node with id $\T(i,j)$ computes the function $\theta(i,j)=\theta(i+1, 2j) \oplus \theta(i+1, 2j+1)$, where $\oplus$ is the binary operation to be performed by node $\T(i,j)$ specified by the function schema. Note that in case of a non-complete binary tree such node $\T(i,j)$ can compute identity function as well if it corresponds to a source node. Also, we have $f_K=\theta(0,0)$. 

Consider an example schema for function $f_4=x_1x_2+x_3x_4$ as shown in the  Figure~\ref{fig:func_schema}a. Note that it is a complete binary tree with height $h=2$ and the source nodes are present at the leaf level. The data packets $x_1,x_2$ are called the operands for the subfunction $x_1x_2$ and are obtained by the identity functions on the respective nodes i.e. $\theta(2,0)=x_1$ and $\theta(2,1)=x_2$. Also, the nodes labelled $\T(1,0),\T(1,1),\T(0,0)$ represent the specific operations that need to be performed on the data. For example, node $\T(1,0)$ performs the multiplication operation and represents the subfunction $\theta(1,0)=\theta(2,0)\times\theta(2,1)=x_1x_2$ for the function $\theta(0,0)=f_4=x_1x_2+x_3x_4$. Next, consider a non-complete binary tree schema for function $f_4=(y_1y_2+y_3)y_4$ as shown in the  Figure~\ref{fig:func_schema}b with height $h=3$. In this case, apart from the leaf level source nodes are present at other levels as well and they perform identity function like $\theta(2,1)=y_3$ and $\theta(1,1)=y_4$. All other nodes perform specific operation  on data like node $\T(0,0)$ combines the subfunctions $(y_1y_2+y_3)$ and $y_4$ to compute function $\theta(0,0)=f_4=(y_1y_2+y_3)y_4$.
\begin{figure}[t]
    \centering
        \resizebox*{0.7\linewidth}{!}{
    \input{func_schema.tex}
}
    \caption{Computation schema for function (a) $f_4 = x_1x_2 + x_3x_4$ (b) $f_4=(y_1y_2+y_3)y_4$.}
    \label{fig:func_schema}
\end{figure}
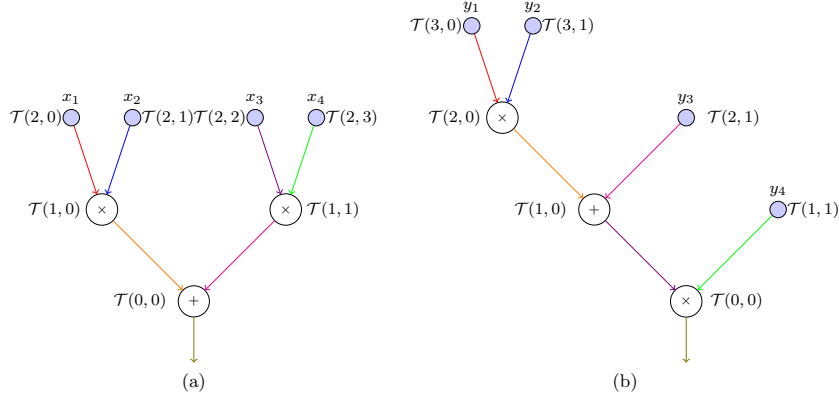

\subsection{In-network Computation Scenarios}
\label{subsec:schemes}
Now, for the given function schema i.e. the binary tree schema let us look at the two computation scenarios by which we can compute such functions. These two scenarios differ based on whether the node at which a specific intermediate value is to be computed is known in advance or not.

\paragraph{In-network computation scenario 1: The Fixed model}
In this model, any subfunction of computation schema $\T$ can be performed only at a specific node in the network i.e. we are given a mapping $\phi: \T \mapsto V.$ We assume $\phi$ is a one-to-one mapping. The schema node id $\phi^{-1}(u)$ is hard-wired into $u \in V$ at the time of deployment of the network so that $u$ knows the specific operation which it has to perform along with the data identifiers of the operands that are the arguments of that operation. Every node $u$ that is in the codomain of $\phi$ maintains two queues; one for storing the data operands of the operation specified by $\phi^{-1}(u)$, we call it  $\mathcal{C}_t(u)$, and another for data
transmission, namely $Q_t(u)$. Once packets of both data operands are received in $\mathcal{C}_t(u),$ node $u$ performs the operation defined by $\phi^{-1}(u)$ and stores the generated data packet in $Q_t(u).$ This increases the data transmission queue size by one. Packets other than these operands are directly stored in $Q_t(u)$ for future transmissions. 

Consider the schema of Figure~\ref{fig:func_schema}a and the network of Figure~\ref{fig:example_fixed}a with a fixed mapping of nodes in schema to the nodes of the network. Note, $V_s=\{x_1,x_2,x_3,x_4\}$ is the source set for function $f_4$ which is mapped to schema nodes $\T(2,0),\T(2,1),\T(2,2),\T(2,3)$ respectively and nodes $\T(1,0),\T(1,1),\T(0,0)$ of schema are mapped to nodes $t,v,w$ of the network respectively. So, network nodes $t,v,w$ know the ids of their operands and the function they need to compute and maintain both data operand queue and data transmission queue. On the other hand, network node $u$ (Figure~\ref{fig:example_fixed}b) is not mapped to any node in schema so it maintains only data transmission queue and relays data $x_1,x_2$ rather than performing the operation specified by $\T(1,0)$ on them (see Figure~\ref{fig:example_fixed}c). The operation of $\T(1,0)$ is performed by node $t$ at the end of $5^{th}$ time slot (Figure~\ref{fig:example_fixed}d). Note that the network node $w$ stores the operand $x_3x_4$ in its data operand queue till it receives the other operand to perform the operation specified by $\T(0,0)$ in Figures~\ref{fig:example_fixed}d, \ref{fig:example_fixed}e.

\begin{figure}[t]
    \centering
    \resizebox*{1\linewidth}{!}{
        \input{example_fixed.tex}
    }
    \caption{Example of \fixedalgo\ algorithm for the function schema of Figure~\ref{fig:func_schema}a: (a) Communication network $G$ with $\phi$ mapping (b) Movement of packets for first two time slots. Node $u$ does not perform any operation (c) Movement of packets in $3^{rd}$ time slot (d)  in $4^{th}$ and $5^{th}$ slot (e) in last three slots.}
    \label{fig:example_fixed}
\end{figure}
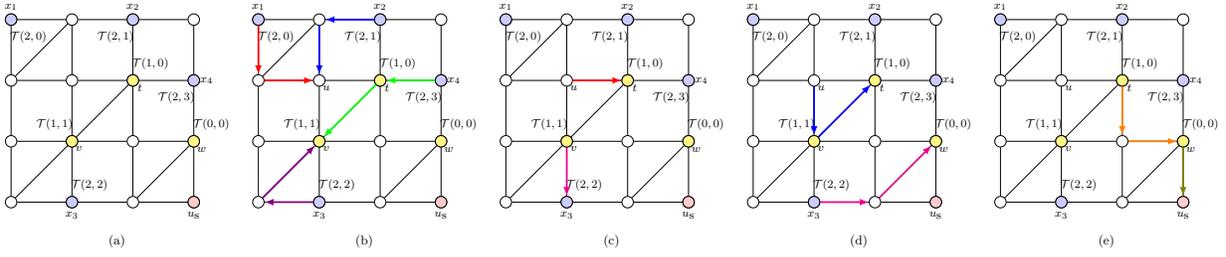

\paragraph{In-network computation scenario 2: The Flexible model} 
In the flexible network model, every node in the network knows the entire schema $\T$ and can perform any operation within it. If the operand data packets for any subfunction of $\T$ are available at any network node $u \in V$ at any time $t$ then it performs the required operation and creates the data packet for the corresponding subfunction. Every node $u$ at time $t$ in the network maintains a single queue $Q_t(u)$ of data packets which it has received (or generated) so far, and has not transmitted (or used for generating any subfunction) yet. For a received packet, if the corresponding operand is available in the queue then they are combined i.e. we perform \emph{coalescence} of packets and the new packet is stored in the queue for future transmission. On the other hand, if the corresponding operand for the received packet is not present in the queue then the received packet is simply stored in the queue without any coalescence.

Consider the schema $\mathcal{T}$ of Figure~\ref{fig:func_schema} and network $G$ of Figure~\ref{fig:example}a. In this model, any node can compute the subfunctions of schema provided it has both the data operands required for computing the subfunction. See Figure~\ref{fig:example}b, where network node $u$ receives data packets $x_1,x_2$ at the end of two time slots and performs the operation of node $\T(1,0)$ generating $x_1x_2.$ Similarly, node $v$ of the network performs the operation of schema node $\T(1,1)$ and $w$ performs the operation of schema node $\T(0,0)$ in the network. All other nodes of the network relay the data packets.

\begin{figure}[t]
    \centering
    \resizebox*{0.8\linewidth}{!}{
        \input{network_example.tex}
    }
    \caption{Example of \flexiblealgo\ algorithm for the function schema of Figure~\ref{fig:func_schema}a: (a) Communication network $G$ (b) Movement of packets for first two time slots. (c) Movement of packets for next four time slots.}
    \label{fig:example}
\end{figure}
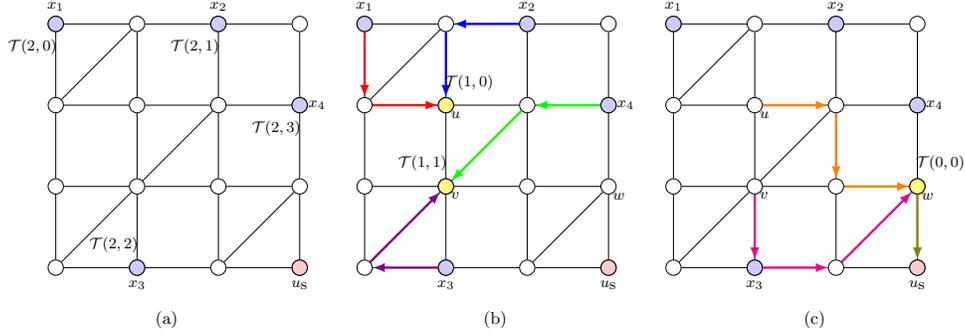

\section{Our Algorithms}
\label{sec:algos}
In this section, we first discuss the random walk-based routing primitive that we use and then we describe our two algorithms for the two in-network computing scenarios discussed before.

\subsection{Routing Scheme Used by Our Algorithms}
\label{subsec:routing}
The nodes in the network follow a slotted time model for communication and only one packet can be transmitted over an edge in a slot.
Our algorithms use a \emph{push} communication model \cite{Mosk-PODC:2006} where at the start of any time step $t$, every node  $u \in V$ selects a node $v$ with probability
\begin{equation}
\probm{u,v} = \begin{cases}
\frac{1}{\degree{u}} & ~\mbox{if } v\in \nbd{u}, \\
0 & ~\mbox{otherwise}
\end{cases}
\label{eq:metropolis_final}
\end{equation}
independent of other nodes and previous time step selections and sends it a randomly chosen data packet from the queue $Q_t(u)$ (if it is not empty). Note that this corresponds to the data packets performing simple random walk on the graph of the network. Next, we present the algorithms for the two in-network computation scenarios.

\subsection{\fixedalgo\ Algorithm}
\label{subsec:fixed_algo}
In the fixed model, we use the algorithm \fixedalgo\  to compute $f_K$ (see Algorithm~\ref{alg:fixed_algo}). In this algorithm, when a node which has a subfunction of $\T$ mapped to it receives a packet it checks to see if the packet is one of the operands of that subfunction. If it is then the node checks its operand queue to see if it has the other operand. If the other operand is available, it combines the two as per the subfunction and moves the combination into its transmission queue. If not, it stores the received packet in the operand queue. If the received packet is not relevant to the subfunction mapped to this node, the node simply places it in its transmission queue for the onward relay. In every time step, the node also chooses a packet uniformly at random and transmits it to a neighbour chosen according to the probability distribution $\probm{\cdot,\cdot}$.  An example run of the \fixedalgo\ algorithm is shown in Figure~\ref{fig:example_fixed}.

\begin{algorithm} [t] 
   \caption{\textbf{Fixed Random-Compute} Algorithm run by node $u$ at time step $t$}
   {\label{alg:fixed_algo}}
   \begin{algorithmic}[1] 
       \REQUIRE Node $u\in V$ that knows the two operands and the operator 
         involved in $\phi^{-1}(u)$ which is a node of $\cal{T}$. If $\phi^{-1}(u)$ is not defined then $\mathcal{C}_t(u)$ is trivially assumed to be empty.
               \IF {$Q_t(u)$ is non-empty}
                   {\STATE $u$ picks $v \in \nbd{u}$ with probability $\probm{u,v}$
                     \STATE Transmits a packet chosen uniformly at random from $Q_t(u)$ to $v$}
                     \ENDIF
               \IF {$u$ receives a packet $p$ in time step $t$} 
               {\IF {$p$ is an operand for $\phi^{-1}(u)$}
               {\IF {$\exists q \in \mathcal{C}_t(u)$ such that $p, q$ are combined at $\phi^{-1}(u)$} \STATE {Combine $p,q$ as per $\phi^{-1}(u)$ and place the combination in $Q_{t+1}(u)$}
               \ELSE \STATE {Place $p$ in $\mathcal{C}_{t+1}(u)$} \ENDIF}
             \ELSE   \STATE {Place $p$ in $Q_{t+1}(u)$} 
               \ENDIF}
               \ENDIF
   \end{algorithmic}
\end{algorithm}

\subsection{\flexiblealgo\ Algorithm} 
\label{subsec:flexible_algo}
We propose the algorithm \flexiblealgo\ for the flexible computation scenario (see Algorithm~\ref{alg:flexi_algo}). This algorithm works by performing combinations allowed by $\T$ opportunistically: when a node receives a new packet from a neighbour it checks its transmission queue to see if that packet can be combined with any packet currently in the queue. If such a combination is possible, it performs it and places the combined value in its transmission queue. As in the case of \fixedalgo\ here too in every time step the node also chooses a packet uniformly at random and transmits it to a neighbour chosen according to the probability distribution $\probm{\cdot,\cdot}$. An example run of the \flexiblealgo\ algorithm is shown in Figure~\ref{fig:example}.
\begin{algorithm}[t] 
   \caption{\textbf{Flexible Random-Compute} Algorithm run by node $u$ at time step $t$}
   {\label{alg:flexi_algo}}
   \begin{algorithmic}[1] 
       \REQUIRE Node $u\in V$ that knows the entire schema $\T.$ 
               \IF {$Q_t(u)$ is non-empty}
                   {\STATE $u$ picks $v \in \nbd{u}$ with probability $\probm{u,v}$
                    \STATE Transmits a packet chosen uniformly at random from $Q_t(u)$ to $v$}
                     \ENDIF
               \IF {$u$ receives a packet $p$ in time step $t$} 
               {\IF {$\exists q \in Q_t(u)$ such that $p, q$ can be combined as per $\T$}  {\STATE Combine $p,q$ as per $\T$ and place the combination in $Q_{t+1}(u)$}
               \ELSE \STATE {Place $p$ in $Q_{t+1}(u)$} \ENDIF}
               \ENDIF
   \end{algorithmic}
\end{algorithm}

\section{Performance Metrics and Our Results}
\label{sec:metrics_results}
In this section, we define our performance metrics, state the main theorems reflecting analysis of those metrics and then, we discuss the consequences of those theorems using examples based on some common network topologies.

Typically in-network computation algorithms are analysed on two metrics: (a) the {\em rate} at which the computation can be carried out given a model of regular data generation~\cite{Mosk-PODC:2006,Kannan-JSAC:2013} and (b) the {\em computation time} or the {\em delay} in computing the function under the assumption that all operands are available at their respective source nodes. These two correspond to notions of the throughput and latency in the in-network computation setting. 

For our rate analysis, we consider the independent Bernoulli data generation model with parameter $\beta$ as described in Section~\ref{subsec:dat_genration_sink}. Note that this data generation model is stochastic and we use techniques from queuing theory to define the best possible function computation rate. The condition for the system to continue processing the data in a regular fashion is known as \emph{stability} in the queueing literature \cite{Gross-BOOK:2008} and is characterised by the fact that the expected size of each queue is finite.  A rate $\beta$ that allows a given algorithm to achieve this condition is termed as the {\em stable rate of computation} for it.

Formally, following Szpankowski~\cite{Szpankowski-TechReport:1989}, we formally define the stable rate of computation of any in-network computation algorithm working under the independent Bernoulli data generation model as follows:
\begin{definition}[Stable rate of computation]
    For any in-network computation algorithm with the data generation as independent Bernoulli process with parameter $\beta$ at all source nodes and with $\lvert V \rvert $ dimensional vector $Q_t^\beta$ representing the state of the network at time $t$, the data rate $\beta$ is said to be stable if the following holds \[ \lim_{t \rightarrow \infty} \prob{\vert| Q^{\beta}_t \vert|_{\infty} < x} = F(x), \mbox{ and }\lim_{x \rightarrow \infty} F(x) = 1, \]
where $\vert| Q^{\beta}_t \vert|_{\infty} = \max \{Q^{\beta}_t(u) : u \in V\}$ and $F(x)$ is the limiting distribution.
\end{definition}

We have dropped superscript $\beta$ from our data transmission queue representation where the rate is understood. Now, we will present the bounds for such a stable rate for the \fixedalgo\ algorithm. 
\begin{theorem}
    \label{thm:rate}
    Given a network $G=(V,E)$ with source set $V_s \subset V$ of $|V_s|=K$ source nodes each receiving independent Bernoulli arrivals with stable rate $\beta$ and a function schema $\T$ to compute a function $f_K$ with binary tree schema $\T$ on $K$ sources, the stable rate $\beta$ of \fixedalgo\ algorithm is given by:
    \begin{equation*}
    \frac{1-\lambda_2}{2\sqrt{3}(K-1)} \left(\frac{d_{min}}{d_{max}}\right)^{1/2}\leq \beta \leq \delta,
    \end{equation*} 
where $\lambda_2$ is the second largest eigenvalue of the transition matrix $\mathcal{P}$ of \fixedalgo,\ $d_{min}$ and $d_{max}$ are the minimum and maximum degree of graph $G$, $\delta_i := \min_{U \subset V: i \in U, \sink \notin U} \sum_{u \in U, v \notin U} \probm{u,v}$ is the mincut between node $i$ and sink $\sink$ and $\delta := \min_{i \in V} \delta_i$ is the min-mincut of $G$.
\end{theorem}
Regarding the stable rate for the flexible model, we do not have any bounds for this model for now. We discuss it in Section~\ref{sec:conclusion} as part of the future work.

\paragraph{Rate results: Discussion and examples}

\begin{table}
\caption[Rate lower bounds]{Rate lower bounds of Fixed Random-Compute for various graphs}
\label{table:rate_bounds}
\begin{center}
\scalebox{0.85}{
\begin{tabular}{lccc}
\hline
\textbf{Graph} & \textbf{$\lambda_2$} & \textbf{Using Theorem~\ref{thm:rate}} & \textbf{Using Exact Rate}\\
\hline\\
Cycle &$1-O(\frac{1}{n^2})$\cite{Levin-BOOK:2009} &$\dfrac{1}{2~\sqrt[]{3}~n^2(K-1)}$ &$\dfrac{1}{(n-1)(K-1)}$\\
Star Graph with sink at centre & &\\
and $\epsilon$ as self loop probability at each node &$0$ &$\dfrac{1}{2~\sqrt[]{3}(n-1)^{1/2}(K-1)}$ &$\dfrac{1-\epsilon}{(K-1)}$\\
Star Graph with\\sink at outer node &$0$ &$\dfrac{1}{2~\sqrt[]{3}(n-1)^{1/2}(K-1)}$ &$\dfrac{1}{2(n-1)(K-1)}$ \\
$m$-dimension\\Hypercube with $n=2^m$ &$1-\frac{1}{\log n}$\cite{Levin-BOOK:2009} &$\dfrac{1}{2~\sqrt[]{3}\log n~(K-1)}$ &$\dfrac{n}{(n+2\log n)(K-1)}$\\
Complete graph &$\frac{-1}{n-1}$ &$\dfrac{n}{2~\sqrt[]{3}(n-1)~(K-1)}$ &$\dfrac{n}{3(n-1)(K-1)}$\\
Random Geometric\\Graph $\big(G(n,r)\big)$ &$1-\Theta(r^2)$\cite{Shah-ICASSP:2009} &$\dfrac{\log n}{2~\sqrt[]{3}~n(K-1)}$ &-\\
\hline
\end{tabular}}
\end{center}
\end{table}
In Table~\ref{table:rate_bounds}, we present two rate bounds for some common network topologies. First we present the lower bounds obtained from Theorem~\ref{thm:rate} and then we present the lower bounds obtained by calculating the exact values using elementary algebra (see~\cite{Gillani-arXiv:2017} for details). We note that for the complete graph the lower bound given by Theorem~\ref{thm:rate} is tight up to small constants, but in the case of the cycle a direct calculation gives us a rate that is higher by $\Theta(n)$.

Our next metric related to the delay or the latency is what we call the average function computation time of an in-network computation algorithm. Now, under our regular data generation model with rate $\beta$ let the computed value of function $f_K$ for the $j^{th}$ data round be $f_K(j):=f_K(x_1(j), x_2(j),
\ldots, x_K(j))$ and let random variable $Y_{f_K(j)}^t$ denote the position of the $j^{th}$ round computed value $f_K(j)$ at the start of time slot $t$ given that the data packets perform simple random walk on $G$. Then, the function computation time of $\ell$ data rounds $\tcomp$ is defined as,
$$\tcomp := \min_t \{Y_{f_K(j)}^t = \sink , \forall~ 1 \leq j \leq \ell\}.$$
So, this represents the earliest time by which the computed value of $\ell$ data rounds is available at the sink. So, now we can define the \emph{average computation time} as follows.
\begin{definition}[Average function computation time]
\label{def:collection_time_k} 
The average function computation time for the network is defined as \[\tbarfk = \lim\limits_{\ell\rightarrow \infty}\dfrac{\tcomp}{\ell}\]
where $\tcomp$ is the the function computation time of the first $\ell$ rounds of data.
\end{definition}

For \fixedalgo\ we prove the following theorem:

\begin{theorem}[Fixed model computation time]
    \label{thm:fixed_time}
    
    Given a network $G=(V,E)$ of $|V|=n$ nodes with set $V_s \subset V$ of $|V_s|=K$ source nodes each receiving independent Bernoulli arrivals with stable rate $\beta$ and a function $f_K$ with binary tree schema $\T,$ for $\beta \leq 1/2$, the average function computation time of $f_K$ using \fixedalgo\ is
    \begin{equation*}
    \tbarfk \leq \alpha \log K \left(\dfrac{1}{\beta}+ \dfrac{h\thits}{1 - c(\beta)}\right)
    \end{equation*}
    where $\alpha>1$ is a constant, $h$ is the height of the binary tree schema $\T$, $c(\beta)\in[0,1]$ is a continuous and increasing function of $\beta$ with $c(0)=0$ and $c(\beta) \rightarrow 1$ as $\beta \rightarrow \beta^*$ where $\beta^*$ is the critical rate below which data rates are stable and above which they are unstable and $\thits=\max\limits_{x,y \in V} \mathbb{E}_x(\tau_y)$ is the worst-case hitting time of simple random walk on $G$. 
\end{theorem}

For \flexiblealgo\ we have the following theorem:
\begin{theorem}[Flexible model computation time]
    \label{thm:flexible_time}

    Given a network $G=(V,E)$ of $n$ nodes with a source set $V_s \subset V$ of $|V_s|=K$ source nodes each receiving independent Bernoulli arrivals with stable rate $\beta$ and a function $f_K$ with binary tree schema $\T,$ for $\beta \leq 1/2$, the average function computation time of $f_K$ using \flexiblealgo\ is
    \begin{equation*}
    \tbarfk \leq \hat{\alpha}\log K \left(\dfrac{1}{\beta}+ \dfrac{h\tmix^G}{1 - c(\beta)}\left(\log^{2}{K}+\dfrac{n}{\nu\log n}\right)\right)
    \end{equation*}
    where $\hat{\alpha}>1$ is a constant, $h$ is the height of the binary tree schema $\T$, $c(\beta)\in[0,1]$ is a continuous and increasing function of $\beta$ with $c(0)=0$ and $c(\beta) \rightarrow 1$ as $\beta \rightarrow \beta^*$ where $\beta^*$ is the critical rate below which data rates are stable and above which they are unstable. Here, $\tmix^G$ is the mixing time of simple random walk on $G$,  $\nu =\sum_{v \in V}{(\degree{v}^2)}/{d^2n}$ where, $\degree{v}$ is the degree of any vertex $v \in V$ and $d = 2m/n$ is the average vertex degree. 
\end{theorem}

\paragraph*{Computation time results: Discussion and examples}
Before comparing the two models, let us look at an important parameter $\nu$ in the flexible model latency result  which represents the variability in the vertex degree. Recall $\nu =\sum_{v \in V}{(\degree{v}^2)}/{d^2n}$ where, $d$ is the average vertex degree. For regular graphs like cycle $\nu=1$ whereas for skewed graphs like star $\nu=O(n)$. This parameter affects the flexible model latency as we can see the star graph has a much lower latency then the cycle graph (see Table~\ref{table:comptime_fixed_flexible}). This is because in the skewed graphs like star higher degree nodes allow more coalescences which are critical for our analysis, whereas it is not so in the regular graphs as all nodes are similar.
 
\begin{table}
\caption[Comparison of average function computation time]{Comparison of the computation time of Fixed and Flexible Random-Compute for various graphs for $\beta<\beta^*$ i.e. $c(\beta)<1$ (All bounds on the mixing and hitting times
appear directly or implicitly in \cite{Aldous-BOOK:2002}.)}
\label{table:comptime_fixed_flexible}
\begin{center}
\begin{tabular}{lcc}
\hline
\textbf{Graph} & \textbf{Fixed Random-Compute} & \textbf{Flexible Random-Compute}\\
\hline\\
Cycle &$O\Big(\dfrac{\log K}{\beta}+ hn^2\log K\Big)$ &$O\Big(\dfrac{\log K}{\beta}+ hn^3\dfrac{\log K}{\log n}\Big)$ \\
Star with sink at centre &$O\Big(\dfrac{\log K}{\beta}+ hn \log K\Big)$ &$O\Big(\dfrac{\log K}{\beta}+ h\log^3K\Big)$ \\
Hypercube &$O\Big(\dfrac{\log K}{\beta}+ hn\log K\Big)$ &$O\Big(\dfrac{\log K}{\beta}+ hn  \log \log n\log K\Big)$ \\
Rand. $r-$reg. &$O\Big(\dfrac{\log K}{\beta}+ hn \log K \Big)$ &$O\Big(\dfrac{\log K}{\beta}+ hn \log K \Big)$ \\
Torus $(d=2)$ &$O\Big(\dfrac{\log K}{\beta}+ hn\log n \log K \Big)$ &$O\Big(\dfrac{\log K}{\beta}+ h n^2\dfrac{\log K}{\log n}\Big)$ \\
Torus $(d>2)$ &$O\Big(\dfrac{\log K}{\beta}+ hn \log K\Big)$ &$O\Big(\dfrac{\log K}{\beta}+ hn^{\frac{d+2}{d}}\dfrac{\log K}{\log n}\Big)$ \\
Rand. geometric graph &- &$O\Big(\dfrac{\log K}{\beta}+ hn^2\dfrac{\log K}{\log n}\Big)$ \\
\hline
\end{tabular}
\end{center}
\end{table}
In Table~\ref{table:comptime_fixed_flexible}, we compare the computation time of Fixed and Flexible Random-Compute for various graphs given that $\beta<\beta^*$ i.e. $c(\beta)<1$ where $c(\beta)\in[0,1]$ is a continuous and increasing function of $\beta$. The given condition ensures that the queues at the nodes are finite and function computation is carried in stable manner. 

Consider Rand $r$- reg graph, in this case the performance of both the algorithms is same whereas in the cycle graph \fixedalgo\ is better than the \flexiblealgo\ but in the star graph latter dominates the former. So, the upper bounds in the table do not clarify whether \flexiblealgo\ is better from a computation time point of view or \fixedalgo. Let us now analyse two models based on their working. Under the fixed schema, two operands that can be combined may meet at some node but may not be combined since the node at which they meet is not the designated node where they are to be combined. This is a major disadvantage that the flexible schema does not have: in the flexible schema the first time two combinable operands meet they are combined. However, the fixed model counterbalances this disadvantage with a significant advantage: since any node designated to combine two operands knows which specific operands it is to combine, it can recognise and store the first of the operands separately and wait for the second operand to arrive. The flexible model cannot do this and must perform opportunistic combinations. It is an open problem to try and establish if the computation time of \fixedalgo\ is stochastically dominated by the computation time of \flexiblealgo\ or vice-versa. 

In next section, we will discuss the proofs of theorems discussed in this section.

\section{Proofs}
\label{sec:proofs}
In this section, we present the detailed proofs of all the theorems. First, we prove the bounds on the stable rate of computation for the fixed model. Then, we give an upper bound on the average function computation time of the fixed and flexible model in terms of the basic random walk primitives: the hitting time and the mixing time respectively.

\subsection{Rate Analysis}
\label{subsec:rate}
Szpankowski defines the stable data rate as the rate which ensures that the queues at all nodes in $V$ are bounded (see Section~\ref{sec:metrics_results}). Now, we discuss the proof of the theorem giving bounds for such stable rate of computation for \fixedalgo\ algorithm. We prove this theorem using an existing rate result of \cite{Gillani-arXiv:2017} and \cite{Banerjee-QS:2012}.

\begin{proof}[Proof of Theorem~\ref{thm:rate}]
First, we will prove a lower bound on the stable rate of the \fixedalgo\ . For this we will use the lower bound on the rate of data collection by an algorithm similar to ours. In an earlier work \cite{Gillani-arXiv:2017} authors propose an algorithm (called Random-Collect) collecting data packets from source nodes to a designated node called sink. In their setting, there is no combination of data packets and all the source data packets travel till they reach the sink. Random-Collect algorithm is a generalised algorithm where data packets perform random walk on the network and the specific type of random walk would depend on the transition probability of moving from one node to another like the transition probability given by Eq. \ref{eq:metropolis_final} allows packets to move in a simple random walk fashion. 
We will use the following simple random walk specific result from \cite{Gillani-arXiv:2017} which proves a lower bound on the stable data rate in that paper.

\begin{lemma}[Network throughput lower bound \cite{Gillani-arXiv:2017}]
\label{thm:rate_bounds}
For a given graph $G=(V,E)$ with $|V|=n$ nodes, source set $V_s \subseteq V \setminus \{\sink\}$ with $|V_s|=s$ data sources, each generating data as independent Bernoulli arrivals with stable data arrival rate $\beta$, and a single sink, $\sink$, we have that 
\begin{equation}
\beta \geq \left(1 - \lambda_2\right)\sqrt{\frac{d_{min}}{d_{max}~2s(s + 1)}}
\label{eq:lambda_lower_bounds_gen}
\end{equation}
where $\lambda_2$ is the second largest eigenvalue of transition matrix $\mathcal{P}$ of simple random walk on graph $G$, $d_{min}$ and $d_{max}$ are the minimum and maximum degree of nodes of the graph $G$ respectively. These results hold for $\beta < \beta^*$, where $\beta^*$ is the critical rate below which data rates are stable and above which they are unstable.
\end{lemma}

Recall that the schema to compute the function $f_K$ on $K$ sources is given by a binary tree $\T$ to \fixedalgo\ algorithm. Thus, there are $K-1$ coalescences of data packets taking place while computing the function. Recall that in \fixedalgo\ algorithm any coalescence, say $\hat{a}$ of data packets $x_i,x_j,$ can happen only at a particular node $u$ in the network and this is given by a mapping $\phi$ from $\T$ to the vertex set $V$ of the network $G.$ In this case, the node $u$ acts like a source to the data packets generated by the coalescence $\hat{a}$ which is used for the next level coalescence. We restrict the coalescence $\hat{a}$ which happens at node $u.$ This reduces the problem to the collection $x_i,x_j$ by node $u.$ By Eq.~\ref{eq:lambda_lower_bounds_gen}, the rate of collection of $x_i,x_j$ by $u$ using \fixedalgo\ is $\beta_{ij} \geq \frac{1-\lambda_2}{2\sqrt{3}}\left(\frac{d_{min}}{d_{max}}\right)^{1/2}.$ 
	Note that while achieving the rate $\beta_{ij}$, \fixedalgo\ might be using some edge of $G$ at its full capacity. There are $(K-1)$ such coalescences happening in $G$ at the same time thus the rate achieved by each of them without exceeding any edge capacity is $\beta_{ij} \geq \frac{1-\lambda_2}{2\sqrt{3}(K-1)}\left(\frac{d_{min}}{d_{max}}\right)^{1/2}.$ Thus, the stable rate of \fixedalgo\ is lower bounded by $\beta \geq \frac{1-\lambda_2}{2\sqrt{3}(K-1)}\left(\frac{d_{min}}{d_{max}}\right)^{1/2}.$

We now turn to the upper bound. 
Authors in \cite{Banerjee-QS:2012} study the rate of function computation when the function $f_K$ is a \emph{divisible} function. A function $f_K$ is said to be divisible if for any partition of its source data set, $f_K$ can be computed by computing a local operation on any set and aggregating the result \cite{Banerjee-QS:2012}. In other words, they study functions whose source data packets can be combined in any order. In our case, source data packets of $f_K$ can be combined only by the order given by schema $\T.$ Thus, the rate achieved in \cite{Banerjee-QS:2012} is an upper bound on the rate achieved by our algorithm. Authors in \cite{Banerjee-QS:2012} show that maximum stable rate for a divisible function is the \emph{min-mincut} of the network graph $G.$ We define the equivalence of min-mincut for the transition matrix of our algorithm (transition probability given by Eq.~\ref{eq:metropolis_final}) for $G$ as follows: recall that $\sink$ is the sink node in $G$ where the final computed function is collected and let $i$ be any node in the network. Then, the mincut between node $i$ and sink $\sink$ is defined as:
\begin{equation*}
\delta_i := \min_{U \subset V: i \in U, \sink \notin U} \sum_{u \in U, v \notin U} \probm{u,v}.
\end{equation*}
Further, the min-mincut of $G$ is defined as $\delta := \min_{i \in V} \delta_i.$ Following the same proof technique as that of Theorem~1 of \cite{Banerjee-QS:2012} we get that the maximum stable rate $\beta$ achieved by \fixedalgo\ algorithm is upper bounded by $\delta.$

\end{proof}

\subsection{Computation Time Analysis} 
\label{subsec:computation_time}
Next, we discuss the latency in function computation or the average time taken to compute the function under the two given models: fixed and flexible. We start with the proof of theorem dealing with the analysis of the computation time for the fixed model and then, present the proof of the flexible model computation time theorem. In both proofs, we first find the expected maximum time taken by $\ell$ data packets to appear at each source node and then, we find the expected time taken to compute the function for such $\ell$ rounds of data generated. Since, all data rounds are independent we find this time by analysing the single round computation time i.e. the time taken to compute the function value for the $i^{th}$ data round comprising of $|V_s|=K$ data packets with round identifier $i$. For the fixed model, we find the single round computation time by recursively reducing our problem to the hitting event of data packets from two nodes onto a designated third node and then we present an upper bound on the time taken for one such event using the techniques and queueing delay analysis of \cite{Gillani-arXiv:2017}. For the flexible model, single round computation time proof involves a major modification of the proof of Cooper et. al.'s main theorem~\cite{Cooper-SIAM:2013} and also uses techniques from~\cite{Gillani-arXiv:2017} to handle queueing delay.
\begin{proof}[Proof of Theorem~\ref{thm:fixed_time}]
Given a graph $G$ where each node in source set $V_s$ receives independent Bernoulli arrivals with stable rate $\beta$. Each data round generated has a total of $|V_s|=K$ data packets and each round has its appearance and computation time in the network. Instead of individually finding such times for each round, we will proceed our analysis by finding the expected maximum time by which $\ell$ rounds of data arrival have happened and then after this time, we find the expected function computation time assuming all $K\ell$ packets have appeared.  

Let $\tappk$ be the appearance time of $\ell$ rounds of data arrival on each node using \fixedalgo\ algorithm for communication. Let $f$ be a hypothetical node where we assume packets reside before arriving at the network nodes, then, $\forall u \in V_s$ and $1 \leq j \leq \ell$, we have  $X_{u,j}^0 = f$, where $X_{u,j}^t$ is the random variable denoting the position of $j^{th}$ round data packet of node $u$ at the start of time slot $t$. If the $j^{th}$ data item appears at node $u$ at time $t'$ then $X_{u,j}^t = f$ for all $t < t'$ and $X_{u,j}^{t'} = u$. With this notation we can define appearance time as: $\tappk = \min_t\{ X_{u,j}^t \in V : 1\leq j \leq \ell, u\in V_s\}$. So, $\tappk$ is the earliest time when $\ell$ packets have appeared at each node. Now, we find the expected appearance time of $\ell$ rounds. 

Let $A$ be the event that a source node $u \in V_s$ did not receive $\ell$ arrivals in time $t_u$. Consider, $t_u \geq 2\ell$ and $\beta\leq1/2$. So, we have
\begin{equation}
\prob{A} = \sum_{i = 1}^{\ell}\binom{t_u}{\ell - i} \beta^{\ell - i} (1 - \beta)^{t_u - (\ell - i)}
\label{eq:prob_eq_event_A}
\end{equation}
As, $t_u \geq 2\ell$ , $\binom{t_u}{\ell - i} \leq \binom{t_u}{\ell - (i - 1)} \forall i:1 \leq i -1 \leq \ell$ also, as $\beta \leq 1/2$, so Eq.~\eqref{eq:prob_eq_event_A} can be written as, $\prob{A} \leq \ell\binom{t_u}{\ell}(1 - \beta)^{t_u}$. Now, let $t_u = w\ell$ for $w > 1$, so using the fact that $\binom{t_u}{\ell} \leq (\frac{et_u}{\ell})^\ell$ previous equation can be rewritten as, $\prob{A} \leq \ell (ew (1 - \beta)^w)^\ell$.
So, we can easily find $w$ such that $ew (1 - \beta)^w \leq 1/e$. Note that the value of $w$ will be only  $\beta$ dependent as $w=\frac{b}{\beta}$ for some constant $b > 1$.
So, for a node $u \in V_s$, $\prob{t_u \geq w\ell} \leq \ell e^{-\ell}$. Since, all random walks of data packets from all $K$ nodes are independent of each other, considering worst-case analysis of all nodes we have,
\begin{equation}
\prob{\exists u: t_u \geq w\ell} \leq K\ell e^{-\ell}
\label{eq:whp_eq of_event_A}
\end{equation}
Now, we need to find the expected value of the maximum appearance time of packets at $K$ nodes let it be $\hat{t}_u=\{t_u:\exists u ~t_u \geq w\ell\}$. We have,
\begin{align*}
\ex{\hat{t}_u - \log eK}&\leq\sum_{i=1}^{\infty} \prob{\hat{t}_u- \log eK \geq i}\\
&\leq w\sum_{j=1}^{\infty} \prob{\hat{t}_u\geq w(\log eK + j)}\\
&\leq w\sum_{j=1}^{\infty}K(\log eK + j)e^{-(\log eK + j)}\\
&\leq w\log eK
\end{align*}
So, we get 
\begin{equation*}
\ex{\hat{t}_u} \leq \log eK + w \log eK \leq w\log eK 
\end{equation*}
So, we have the expected maximum appearance time of data packets at $K$ source modes as at most $w\log eK$ where $w = \frac{b}{\beta}$ for some constant $b > 1$. Now considering the worse-case analysis, the maximum time it takes for complete $l$ rounds of data packets to appear at all $K$ source nodes is,
\begin{equation}
\tappk \leq \ell \log eK\frac{b}{\beta}
\label{eq:tau_app}
\end{equation}

Note that after $\tappk$ time all nodes have $\ell$ arrivals. Let $\tcompk{\ell}$ be the expected function computation time of $\ell$ data rounds given the data packets have arrived at all source nodes. Now, we know because of the appearance time of rounds, random walks of data packets across the rounds are independent, so we can write $\tcompk{\ell} \leq \ell\tcompk{1}$.

Now, to find the single round function computation time $\tcompk{1}$ we perform the following analysis. Recall that the computation schema for function $f_K$ is given by a binary tree $\T$ with height $h$ where $\log K \leq h <K$. If $\T_i$ is a complete subtree of binary tree $\T$ with its root at level $i$, then let $\tau_{\T_i}$ be the coalescence time of $\T_i$. We define the coalescence time with respect to three vertices $x, y,z \in V$ as $ \tau_{(x,y;z)} = \max\{\tau_{(x,z)}, \tau_{(y,z)}\}$ where $\tau_{(a,b)}$ is the time taken for a random walk starting from $a$ to first hit $b$, so for $\T_2$ we have $\tau_{\T_2} = \tau_{(\phi(1,0), \phi(1,1); \phi(0,0))}$.

	In general, considering the worse case scenario we have 
	
	$\tau_{\T_i} = \tau_{(\phi(1,0), \phi(1,1); \phi(0,0))} + \max\{\tau_{\T_{i-1,0}}, \tau_{\T_{i-1,1}}\},$
	where $\T_{i-1,0}$ and $\T_{i-1,1}$ are the two subtrees of the root with height $i-1$.
	Applying the recursion again we get 
	$\tau_{\T_{i-1,0}} = \tau_{(\phi(2,0), \phi(2,1); \phi(1,0))} + \max\{\tau_{\T_{i-2,00}}, \tau_{\T_{i-2,01}}\},$
	and
	$\tau_{\T_{i-1,1}} = \tau_{(\phi(2,2), \phi(2,3); \phi(1,1))} + \max\{\tau_{\T_{i-2,10}}, \tau_{\T_{i-2,11}}\}.$
	Observing that for any four random variables $W, X, Y, Z$, $\max\{W + X, Y + Z\} \leq \max\{W, Y\} + \max\{X,Z\}$ we get
	\begin{align*}
	\tau_{\T_i} \leq& \tau_{(\phi(1,0), \phi(1,1); \phi(0,0))} + \max \{  \tau_{(\phi(2,0), \phi(2,1); \phi(1,0))}, \tau_{(\phi(2,2), \phi(2,3); \phi(1,1))}\}\\
	& + \max\{\tau_{\T_{i-2,00}}, \tau_{\T_{i-2,01}},\tau_{\T_{i-2,10}}, \tau_{\T_{i-2,11}}\}.
	\end{align*}
	
	Therefore we get 
	\begin{equation}
	\tau_{\T_h} \leq \sum_{i=1}^{h} \max_{0\leq j \leq 2^{i-1}-1}\{ \tau_{(\phi(i,2j), \phi(i,2j+1); \phi(i-1,j)) }\}.
	\label{eq:tau_k_rec}
	\end{equation}
	
	Note if we have a skewed binary tee schema with $h=K-1$ then in the above equation many variables of form $\tau_{(\cdot,\cdot)}$ will not be defined whereas for the complete binary tree schema with $h=\log K$ all such variables are present. Now, from the definition of coalescence time we know that each coalescence consists of two hitting time events. So, for each coalescence of form $\tau_{(x,y;z)}$ let us denote two random variables $\tau_{(x,y;z)}^1$ and $\tau_{(x,y;z)}^2$ for the two hitting time events. Let set $B$ be defined as $B:=\{\tau_{(\phi(i,2j), \phi(i,2j+1); \phi(i-1,j)) }^k:1\leq i\leq h, 0\leq j \leq 2^{i-1}-1 \text{ and } k=1,2\}$.  Note that there are always $2(K-1)$ random variables of form $\tau_{(x,y;z)}^k$ for $k=1,2$ on the right hand side of Eq.~\eqref{eq:tau_k_rec}. Now, let the hitting time for the $i^{th}$ out of $2(K-1)$ hitting event in the Markov chain associated with the simple random walk (with {\em no delay}) be $\tau_{(x,y;z)}^i$ for node $x,y,z \in V$. So, we have $\mathbb{E}[\tau_{(x,y;z)}^i] \leq \thits$ where $\thits=\max\limits_{x,y \in V} \mathbb{E}_x(\tau_y)$ is the worst-case hitting time of the simple random walk starting from any node of the graph. By Markov's inequality $\prob{\tau_{(x,y;z)}^i \geq e\thits} \leq \frac{1}{e}$. Now, consider the probability of the random walk not hitting $z$ in $we$ times the worst-case hitting time i.e. we consider $w~e\thits$ time and divide it into $w$ slots $(l=1,2,\cdot\cdot\cdot,w)$ of $e\thits$ each. By the Markov property of random walks, we know that the random walks in each of these slots are independent. Also, since we have used the worst case hitting time for bounding the probability of one slot using Markov's inequality, this bound will hold true for any starting vertex at the start of any slot $l$. So, we have $\prob{\tau_{(x,y;z)}^i\geq w~e\thits} \leq \prod_{l=1}^w \frac{1}{e}\leq \frac{1}{e^w}$. Now, let $\tau_{max}=\max_{\tau\in B}\tau_{(x,y;z)}^i$, then $\prob{\tau_{max} \geq w~e\thits} \leq 2(K-1)e^{-w}$.
	Now, we have
\begin{align*}
\ex{\tau_{max} - \log eK} &\leq \sum_{i=1}^{\infty} \prob{\tau_{max}-\log eK \geq i}\\
&\leq e\thits \sum_{j=1}^{\infty} \prob{\tau_{max}\geq (\log eK+j)e\thits}\\
&\leq e\thits \sum_{j=1}^{\infty} 2(K-1)e^{-\log eK} e^{-j}\\
&\leq 2\thits\frac{K-1}{K(e-1)}
\end{align*}
So, 
$$\ex{\tau_{max}}\leq \log eK + 2\thits\frac{K-1}{K(e-1)} \leq \thits \log eK$$ as both $\log eK>0$ and $\thits>0$.
So without any queueing delay, the expected time taken by the data packets involved in $2(K-1)$ coalescences (hitting events) to hit their designated nodes is $\thits \log eK$. This analysis is done assuming there is only one packet in the queue of any node at any time. 
	
  Now, we analyse the delay caused due to more than one packets in the queue. Consider a packet generated at node $p \in V_s$, let us call it $p$ after its generating node. Let $X_{p,1}^t=X_p^t$ denote the position of packet $p$ at time $t$, we drop the subscript denoting the round number as we are dealing with one round function computation.
Now consider for any node $u \in V$,
\begin{align}
&\prob{\mbox{Packet } p \mbox{ is delayed at } u \mbox{ at time }t} =\nonumber \\
& \sum_{w\geq2} \prob{\mbox{Packet } p \mbox{ is not picked at } u \mbox{ at time } t \cap Q_t^\beta(u)=w \mid X_p^t=u}\prob{X_p^t=u}
\label{eq:cond1}
\end{align}
We know, 
\begin{align}
&\prob{\mbox{Packet } p \mbox{ is not picked at } u \mbox{ at time } t \cap Q_t^\beta(u)=w \mid X_p^t=u}= \nonumber\\
&\frac{w-1}{w}\frac{\prob{Q_t^\beta(u)=w \cap X_p^t=u}}{\prob{X_p^t=u}}
\label{eq:cond2}
\end{align}
As, $\prob{A\cap B}\leq \prob{A}$, we can write $\prob{Q_t^\beta(u)=w \cap X_p^t=u}\leq\prob{Q_t^\beta(u)=w}$. Using this result in Eq.~\ref{eq:cond2} we get
\begin{align}
&\prob{\mbox{Packet } p \mbox{ is not picked at } u \mbox{ at time } t \cap Q_t^\beta(u)=w \mid X_p^t=u}\leq \nonumber\\
&\frac{w-1}{w}\frac{\prob{Q_t^\beta(u)=w}}{\prob{X_p^t=u}}\leq \frac{\prob{Q_t^\beta(u)=w}}{\prob{X_p^t=u}}
\label{eq:cond3}
\end{align}
Using Eq.~\ref{eq:cond3} in Eq.~\ref{eq:cond1}, we get
\begin{equation}
\prob{\mbox{Packet } p \mbox{ is delayed at } u \mbox{ at time }t} \leq \sum_{w\geq2} \prob{Q_t^\beta(u)=w}=\prob{Q_t^\beta(u)\geq2}
\label{eq:cond4}
\end{equation}
Now, we have
\begin{align}
&\prob{\mbox{Packet } p \mbox{ is delayed at time }t}=\nonumber\\
& \sum_{u \in V} \prob{\mbox{Packet } p \mbox{ is delayed at } u \mbox{ at time }t\mid X_p^t=u}\prob{X_p^t=u}
\label{eq:condsum}
\end{align}
Now, let $c_t(\beta)=\max_{u \in V} \prob{Q_t^\beta(u)\geq2}$ be the maximum delay probability at time $t$ over all nodes in set $V$ for given data rate $\beta$. So, at stationarity the maximum delay probability converges to $c(\beta)$ i.e., $c(\beta)=\lim_{t\rightarrow \infty}c_t(\beta)$. $c(\beta)$ is a continuous and increasing function of $\beta$ (see Claim~1 \cite{Gillani-arXiv:2017} for details) with $c(0)=0$ and $c(\beta) \rightarrow 1$ as $\beta \rightarrow \beta^*$ where $\beta^*$ is the critical rate below which data rates are stable and above which they are unstable. We know that any packet gets delayed at time slot $t$, due to queue at a node $u,$ because it is not picked for transmission in that slot among all the packets in the queue. Thus, the probability of a packet not being delayed by a node in a given time slot $t$ is $ 1 - c(\beta)$. So, combining the queueing delay with the non-delayed hitting time events, we have the expected number of steps taken by the packets involved in $2(K-1)$ coalescences (hitting events) to hit their designated nodes as $\frac{\thits \log eK}{1 - c(\beta)}$. This expected time is computed without considering any dependence among the hitting events. Now, we use this result in Eq.~\eqref{eq:tau_k_rec} to capture the worse case relation between different levels of binary tree schema and hence the associated hitting time events. So, we have $C(K) = \tau_{\T_h}\leq h \dfrac{\thits \log eK}{1 - c(\beta)}$.

So, we have one round function computation time i.e., assuming each source node has generated only one data packet as 
$\tcompk{1} = C(K) =  h \dfrac{\thits \log eK}{1 - c(\beta)}$, where $c(\beta)\in[0,1]$ is a continuous and increasing function of $\beta$ representing the maximum delay probability in the limit over all nodes in $V$, $\thits$ being the worst-case hitting time of simple random walk on $G$. So,
\begin{equation}
\tcompk{\ell} \leq \ell\tcompk{1} = \ell\left(h \dfrac{\thits \log eK}{1 - c(\beta)} \right).
\label{eq:clearnce_time_k}
\end{equation}
Recall, $\tcomp$  is the expected function computation time of $\ell$ rounds of data arrival and so, we can write 
\begin{equation}
\tcomp = \tappk + \tcompk{\ell}.
\label{eq:tau_col_k_initial}
\end{equation}
Now, using the results from Eq.s \eqref{eq:tau_app} and \eqref{eq:clearnce_time_k}, we have
\begin{equation}
\tcomp \leq \ell \log eK \dfrac{b}{\beta}+ \ell\left(h \dfrac{\thits \log eK}{1 - c(\beta)} \right)
\label{eq:tcol_k_actual}
\end{equation}
So, we have,
\begin{equation}
\tbarfk = \lim_{\ell\rightarrow \infty}\dfrac{\tcomp}{\ell} \leq \log eK \left(\dfrac{b}{\beta} + h \dfrac{\thits}{1 - c(\beta)}\right)
\label{eq:tcol_k}
\end{equation}
where $b>1$ is a constant and $\thits$ is the worst-case hitting time of simple random walk on $G$.
\end{proof}

For the flexible model computation time we have the following theorem proof.

\begin{proof}[Proof of Theorem \ref{thm:flexible_time}]
Now, for the flexible model again we are given a graph $G$ with source nodes in set $V_s$ receiving independent Bernoulli data arrivals at stable rate $\beta$. We will prove this theorem on same lines as that of fixed model computation time proof. We will first find the expected maximum time by which $\ell$ rounds of data arrival have happened and then after this time, we find the expected function computation time assuming all $K\ell$ packets have appeared given that $|V_s|=K$.  

Recall that $\tappk = \min_t\{ X_{u,j}^t \in V : 1\leq j \leq \ell, u\in V_s\}$ is the appearance time of $\ell$ rounds of data arrival on each node using \flexiblealgo\ algorithm for communication, where $X_{u,j}^t$ is the random variable denoting the position of $j^{th}$ round data packet of node $u$ at the start of time slot $t$. Using the same analysis as in the proof of Theorem~\ref{thm:fixed_time}, we get the expected time for complete $\ell$ data rounds to be generated at all $K$ source nodes as 
\begin{equation}
\tappk \leq \ell \log eK\frac{b}{\beta}
\label{eq:tau_app_2}
\end{equation}
where $b>1$ is a constant. Note that after $\tappk$ time all nodes have $\ell$ arrivals. Let $\tcompk{\ell}$ be the expected function computation time of $\ell$ data rounds given the data packets have arrived at all source nodes. Now, we know because of the appearance time of rounds, random walks of data packets across the rounds are independent, so we can write $\tcompk{\ell} \leq \ell\tcompk{1}$.

Now, we perform the following analysis to find the value of $\tcompk{1}$. Recall that the communication network is represented by graph $G=(V,E)$ with sources in set $V_s \subseteq V.$ The schema to compute function $f_K$ is a binary tree of height $h$ where $\log K \leq h <K$ thus has $h+1$ levels. We label the levels in the schema starting from the sink. Leaf level i.e. level $h$ simply acts as source of data for the function to be computed, hence nodes at this level perform identity function $\theta(i,j)=x_{j+1}$ for $0 \leq j < 2^i$. However, in case of a non-complete binary tree source nodes may be present at other levels as well, so network nodes in $V_s$ always perform identity function. Nodes other than the source nodes with id $\T(i,j)$ at all other levels i.e.  level $0 \leq i < h$ compute the function $\theta(i,j)=\theta(i+1, 2j) \oplus \theta(i+1, 2j+1)$, where $\oplus$ is the binary operation to be performed by node $\T(i,j)$ specified by function schema and $0 \leq j < 2^i$. This subfunction computation can be seen as the coalescence of the data packets of the operands. For example, in Figure~\ref{fig:func_schema}a, there are three levels with level $2$ providing four data operands, level $1$ performing two coalescences on the respective operands and level $0$ performing one coalescence. So, because of the in-network computation paradigm each level $i+1$ acts as source for the level $i$.
Now, we will first analyse the computation time for $(h-1)^{th}$ level and replicate it for other levels. If the time required to complete $j^{th}$ level is $C_{k_j}$ then the total time to complete the function computation is $C(K) = \sum_{j=0}^{h -1} C_{k_j} = O(h C_{k_{(h -1)}}).$ For the  simplicity of notations, we remove the subscript $j$ denoting the level number from various notations whenever it is clear from the context.
	
	\paragraph*{Multiple random walks to single random walk.}
	Given $k = 2^i$ data packets from a single level, $i$, of $\T$ let us consider the collection of random walks $\{ X_{i,j}(t) : 0 \leq j < k\}$ executed by these $k$ packets in $G$. We represent this collection as a single random walk on the product graph $Q_k = (V_{Q_k},E_{Q_k})$ where $V_Q = V^k.$ Consider the set 
	\begin{equation}
	S_k := \left\{(v_0, ... , v_{k-1}): v_{2l} = v_{2l+1}:  0 \leq l \leq \frac{k}{2}-1 \right\}. \label{eq:reductionset}
	\end{equation}
	Note that when the random walk on $Q_k$ visits any vertex $S_k$ it is equivalent to all the $k/2$ pairs of walks on $G$ simultaneously coalescing which is what we require for computing the collection of functions $\{ \theta(i,2l) \oplus \theta(i,2l+1) : 0 \leq l \leq 2^{i-1} -1\}$.
	
In other words, each node $\textbf{v} \in V_{Q_k}$ is a $k$-tuple of nodes $\{v_1,\ldots,v_k\}$ where each $v_i \in V.$ Thus, $k$ random walks $X_{u_i}(t)$ on graph $G$ starting from $u_i \in V_s$ can be replaced by a single random walk $X_{\textbf{u}}(t)$ on graph $Q_k$ with starting position $\textbf{u} = (u_1, u_2,\ldots, u_k)$. We further reduce the graph $Q_k$ to a graph $\Gamma_k$ by contracting the set $S_k$ to a single vertex $\gamma_k$ while retaining all other vertices, edges and loops. Thus, the degree of every vertex of $\Gamma_k$ is same as that of $Q_k$ except for the vertex $\gamma_k.$ The degree of $\gamma_k$ is the sum of degrees of all vertices of set $S_k$ and is given by $d_{\gamma_k}.$ Also if $\pi$ and $\hat{\pi}$ are the stationary distributions of a random walk on graphs $Q_k$ and $\Gamma_k$ respectively, then $\hat{\pi}_v = \pi_{v}$ for $v \notin S_k$ and $\hat{\pi}_{\gamma_k} = \pi_{S_k} = \sum_{x \in S_k} \pi_x$.
	
	\paragraph*{Hitting time from stationarity.}
		Recall that the probability of a random walk to move from vertex $u$ to $v$ in graph $G$ is given by $\probm{u,v}$ (defined by Eq.~\ref{eq:metropolis_final}) in our algorithm. Let $\pi$ be the stationary distribution of our random walk in $G$. Let $\mathcal{P}^t[u, \cdot]$ be the probability distribution of the Markov chain associated with the given random walk that begins at state $u$ at time $t$. Then, the distance of this distribution from its stationary distribution $\pi$ at time $t$ is defined as $d(t) := \max_{u\in V} \lvert|\mathcal{P}^t[u, \cdot] - \pi|\rvert_{TV}$. Then, for any $\epsilon > 0$ mixing time $\tmix^G$ of the random walk on $G$ with transition probability given by Eq.~\ref{eq:metropolis_final}, is given by
	\begin{equation}
	\tmix^G=\min\{t:d(t)\leq\epsilon\}.
	\label{eq:mixing_condition}
	\end{equation}
	This measures the time required by the
	Markov chain for the distance to the stationarity to be small. The mixing time for graphs $Q_k$ and $\Gamma_k$ are defined similarly and are denoted by $\tmix^{Q_k}$ and $\tmix^{\Gamma_k}$ respectively. We use the following relation among the mixing times of the three graphs in our proof.
	
	\begin{lemma} (\cite{Cooper-SIAM:2013}, Lemma~2) 
		\label{lm:mixing_time}
		Mixing time of simple random walk with transition probability given by Eq.~\ref{eq:metropolis_final} on graphs $G,Q_k,\Gamma_k$ are given as
		\begin{equation*}
		\tmix^G = O\Bigg(\frac {\log n}{1 - \lambda_2}\Bigg); \tmix^{Q_k} = O(k\tmix^G); \tmix^{\Gamma_k} = O(k\tmix^G), 
		\end{equation*}
		such that
		\begin{equation*}
		\underset{u \in V_F}{max}\lvert| \mathcal{P}^t[u, \cdot] - \pi|\rvert_{TV} \leq 1/{n_F^2} \text{, for any t } \geq \tmix^F
		\end{equation*}
		where $F$ can be any of the graphs $G,Q,\Gamma,$ $n_F = \lvert V_F\rvert$ and $1-\lambda_2$ is the spectral gap of the transition matrix of the simple random walk.
	\end{lemma}
	
	\begin{proof}
		The bound on mixing time of $G$ directly follows from Theorem~12.3 of \cite{Levin-BOOK:2009}. Relation between the mixing times for $Q_k$ and $\Gamma_k$ with that of $G$ follow directly from the proof of Lemma~2 of \cite{Cooper-SIAM:2013}.
	\end{proof}
	
	Let $\mathbb{E}_{\pi}(\tau_v)$ be the expected hitting time of a vertex $v$ starting from the stationary distribution $\pi.$ From Proposition~10.19 of \cite{Levin-BOOK:2009}, we can write
	\begin{equation}
	E_{\pi}(\tau_v) = Z_{vv} / \pi_v,
	\label{eq:expected_hitting}
	\end{equation}
	where,
	\begin{equation}
	Z_{vv} = \sum_{t=0}^{\infty} (\mathcal{P}^t[v,v] - \pi_v).
	\label{eq:z_vv}
	\end{equation}

	Let $M_k(\textbf{u})$ be the time until the first two random walks meet in $Q_k$ when they start from $\textbf{u}.$ In other words, $M_k(\textbf{u})$ is the time to reach the vertex $\gamma_k$ in graph $\Gamma.$ Now, if mixing time $\tmix^{\Gamma}$ of graph $\Gamma$ satisfies Eq.~\ref{eq:mixing_condition}, then
	\begin{equation}
	\mathbb{E}(M_k(\textbf{u})) \leq \tmix^{\Gamma} + (1 + o(1)) \mathbb{E}_{\pi}(\tau_{\gamma_K})
	\label{eq:first_meetingtime}
	\end{equation}
	where $\mathbb{E}_{\pi}(\tau_{\gamma_k})$ is the hitting time of vertex $\gamma_k$ from stationary distribution $\pi$ in graph $\Gamma$. The following lemma follows directly from Lemma~3 of \cite{Cooper-SIAM:2013}.
	
	\begin{lemma}
		\label{lm:bound_zvv}
		Let the spectral gap of the transition matrix $\mathcal{P}[\cdot,\cdot]$ of simple random walk on graph $F$ be $1-\lambda_2$, then for any vertex $v$ of graph $F,$ we have
		\begin{equation*}
		Z_{vv} \leq \frac{1}{1-\lambda_2}
		\end{equation*}
		where $Z_{vv}$ is defined in Eq.~\ref{eq:z_vv}. This lemma holds for all three graphs $G,Q_k$ and $\Gamma_k$. 
	\end{lemma}
	
	Now we prove the following Lemma which will be used to prove the time for single level coalescence.
	
	\begin{lemma}
		\label{lm:probmeasure_gamma}
		Let $G$ be a connected graph with $n$ vertices and $m$ edges and let
		\begin{equation}
		l^* = \text{max}\bigg\{2, \text{min}\bigg\{\frac{n}{\nu}, \log {K}\bigg\}\bigg\},
		\label{eq:k_star}
		\end{equation}
		where  $\nu =\sum_{v \in V}{(\degree{v})^2}/{d^2n}$. Let $k$ be an integer, $2 \leq k \leq l^*$ and $\gamma_k$ be the vertex in graph $\Gamma$ representing the contracted set $S_k.$ Then, there exists a constant $c_k > 0$ such that,
		\begin{equation*}
		\pi_{\gamma_k} = \frac{\degree{\gamma_k}}{(2m)^k} \geq \frac{c_k k\nu}{n}.
		\end{equation*}
	\end{lemma}
	
	\begin{proof}
		Recall the definition of the set $S_k$ which is as follows:
		\begin{displaymath}
		S_k := \left\{(v_0, ... , v_{k-1}): v_{2\ell} = v_{2\ell+1}:  0 \leq \ell \leq \frac{k}{2}-1 \right\}.
		\end{displaymath}
		If $k = 2,$ then, the degree of set $S_k$ is:
		\begin{equation*}
		\degree{S_k} = \sum_{v \in V} (\degree{v})^2 = (2m)^2 \frac{\nu}{n}.
		\end{equation*}
		If $3 \leq k \leq l^*$ and for any $x,y$ such that $0 \leq x < y < k$, we define subsets of $S_k$ as:
		\begin{equation*}
		S_{(x,y)} = \left\{(v_0, ... ,v_{k-1}): v_x = v_y ; y = x + 1, x = 2\ell , 0 \leq \ell \leq \frac{k}{2}-1\right\}.
		\end{equation*}
		We get, $S_k = \underset{0 \leq x < y < k}{\cup}  S_{(x,y)}$, where $\degree{S_{(x,y)}} = (2m)^{k-2}\sum_{v\in V}(\degree{v})^2 = (2m)^k\frac{\nu}{n}$.
		
		By the definition of our function $f_K,$ for $\{x,y\}\neq\{p,q\}$, $\{x,y\} \cap\{p,q\} = \phi$. So, $\degree{S_{(x,y)}} \cap S_{(p,q)}) = (2m)^{k-4}\sum_{u,v\in V}(\degree{u})^2(\degree{v})^2$. 
		Thus by inclusion-exclusion principle,
		\begin{eqnarray} 
		\degree{\gamma_k} = \degree{S_k} & \geq & \underset{\{x,y\}}{\sum}\degree{S_{(x,y)}}-\underset{\{x,y\}\neq\{p,q\}}{\sum}\degree{S_{(x,y)}} \cap S_{(p,q)})\nonumber \\ 
		& \geq & \dfrac{k}{2}(2m)^k\dfrac{\nu}{n} - \dbinom{k/2}{2}(2m)^k\dfrac{\nu^2}{n^2} \label{eq:degree_1}\\
		& \geq & \dfrac{k}{2} (2m)^k \dfrac{\nu}{n}\bigg(1 - \dfrac{k\nu}{4n}\bigg)\label{eq:degree_2}\\
		& \geq & \dfrac{k}{2} (2m)^k \dfrac{\nu}{n} (1 - o(1))\label{eq:degree_3} \\
		& \geq & (2m)^k c_k\dfrac{k\nu}{n}.
        \label{eq:degree_4} 
		\end{eqnarray}

		The factor of $k/2$ in Eq.~\ref{eq:degree_1} is the result of fixed combinations among source nodes because of restricted coalescence defined by function and the combinatorial factor in other term is because of disjoint nature of nodes i.e. $\{x,y\} \cap\{p,q\} = \phi$. The bound in Eq.~\ref{eq:degree_3} follows from Eq.~\ref{eq:degree_2}, by using upper bound on $k$ from Eq.~\ref{eq:k_star}. Then, using bound on $\degree{\gamma_k}$ from Eq.~\ref{eq:degree_4} for $\pi_{\gamma_k}$, we get the desired result.
	\end{proof}
	
	Recall that $\mathbb{E}_{\pi}(\tau_{\gamma_k})$ is the expected time to hit the vertex $\gamma_k$ from stationary distribution $\pi.$ Then by Lemma~\ref{lm:bound_zvv} and Lemma~\ref{lm:probmeasure_gamma} we get,
	\begin{equation}
	\mathbb{E}_{\pi}(\tau_{\gamma_k}) \leq \dfrac{n}{c_k k \nu} \dfrac{1}{1 - \lambda_2}.
	\label{eq:hitting_tau}
	\end{equation}
	
	\paragraph*{Computation time without queueing delay for the simple random walk.}Recall that $M_k(\textbf{u})$ is the time of coalescences of $k \leq l^*$ random walks in $G$ i.e, the time to hit $\gamma_k$ in graph $\Gamma$ which is given by Eq.~\ref{eq:first_meetingtime}. By Lemma~\ref{lm:mixing_time} and Eq.~\ref{eq:hitting_tau} we get,
	
	\begin{eqnarray}
	\mathbb{E}[M_k(\textbf{u})] & \leq & \text{O}(k\tmix^G) + (1 + o(1))\mathbb{E}_{\pi}(\tau_{\gamma_k})\label{eq:firstmeeting_1}\\
	& = & \text{O}\Bigg(\dfrac{1}{1 - \lambda_2}\bigg(k \log{n} + \dfrac{n}{\nu k}\bigg) \Bigg).
    \label{eq:firstmeeting_2}
	\end{eqnarray}
	Eq.~\ref{eq:firstmeeting_2} gives the bound for expected first coalescence time among $k \leq l^*$ data packets, indicating partial computation of $(h -1)^{th}$ level of function schema. Now, for complete computation of $(h -1)^{th}$ level, we will prove that with probability $1 - 1/n^c$, where $c > 0$ is some constant there cannot be a subset of $k = l^*$ data packets which did not coalesce by time $t^*,$ where  $t^* = l^* \log {K}(\tmix^\Gamma + 3 E_{\pi}(\tau_{\gamma_k})).$  To prove this we need the following result.
	
	\begin{lemma}(\cite{Cooper-SIAM:2013}, Lemma 1) 
		\label{lm:prob_notvisit}
		The probability of the event $N^t(u,v)$ such that a random walk starting from $u$ does not visit vertex $v$ in first $t$ time steps is given by
		\begin{equation*}
		\mathbb{P}(N^t(u,v)) \leq e^{-\lfloor t/(T + 3 \mathbb{E}_{\pi}(\tau_v))\rfloor}.
		\end{equation*}
		Here $T = \tmix^G$ is the mixing time of the simple random walk on $G$ and $\mathbb{E}_{\pi}(\tau_v)$ is the expected hitting time of vertex $v$ from stationary distribution $\pi$.
	\end{lemma}
	
	Let $\mathcal{N}(k,\textbf{v})$ be the set of data packets starting from vertex $\textbf{v} = (v_1,\cdots,v_k).$ There are two cases for coalescences of these packets. Either these data packets have coalesced during mixing time $\tmix^\Gamma$, or they have not. For latter case, we can use Lemma~\ref{lm:prob_notvisit} on graph $\Gamma_k$ with vertex $\gamma_k$ and $t =t^*$. This gives us the probability that the data packets have not coalesced by time $t$ which is same as the probability of random walk $\hat{X}_{\textbf{v}}$ on graph $\Gamma_k$ not hitting vertex $\gamma_k$ by time $t.$ So, using Lemma~\ref{lm:prob_notvisit},  $\mathbb{P}(\mbox{No meeting among }\mathcal{N}(k,\textbf{v}) \mbox{ before }t^*)$ is:
	\begin{equation*}
	\mathbb{P}(N^t(\textbf{v},\gamma_k)) \leq  e^{-l^* \log{K}} = K^{-l^*}.
	\end{equation*}
	So,
	\begin{equation*}
	\mathbb{P}(\exists \text{subset of $l^*$ packets which did not coalesce by $t^*$}) \leq \binom{K}{l^*}K^{-l^*} \leq \dfrac{1}{2}.
	\end{equation*}
	The bound of earlier equation is achieved from value of $l^*$ in Eq.~\ref{eq:k_star} and upper bound on binomial coefficient. So, expected number of steps until fewer than $l^*$ data packets remain is at most $ t^* + \frac{1}{2}(2t^*) + \frac{1}{4}(3t^*) + ... = 4t^*$. This is because after every $t^*$ step, number of data packets reduces by half due to coalescing. Now combining the above mentioned result with Eq.~\ref{eq:firstmeeting_2}, we can get the coalescence time for $(h -1)^{th}$ level as,
	\begin{eqnarray*}
		C_{k_{h -1}}^{\scriptsize\mbox{no delay}} &\leq& 4t^* + \mathbb{E}(M_{l^*})\\
		&=& O\Bigg(\dfrac{1}{1 - \lambda_2}\bigg((l^*)^2\log{K}\log{n}+\dfrac{n}{\nu}\log{K}+l^*\log{n}+\dfrac{n}{\nu l^*}\bigg)\Bigg) \\
		&=&  O\Bigg(\dfrac{1}{1 - \lambda_2}\bigg((l^*)^2\log{K}\log{n}+\dfrac{n}{\nu}\log{K}\bigg)\Bigg).
	\end{eqnarray*}
	Now, for $l^* = \log{K}$
	\begin{equation}
	C_{k_{h -1}}^{\scriptsize\mbox{no delay}}= O\Bigg(\dfrac{1}{1 - \lambda_2}\bigg(\log^{3}{K}\log{n}+\dfrac{n}{\nu}\log{K}\bigg)\Bigg).
	\label{eq:firstlevel_coal}
	\end{equation}
	Bound of Eq.~\ref{eq:firstlevel_coal} holds even if $l^* < \log{K}$. Note that this analysis holds true for coalesced packets as well, as after coalescing they are treated in similar way as data packets which originated from source set $V_s$. So, repeating similar analysis for all $h$ levels we get total function computation or coalescence time without any queueing delay as:
	\begin{eqnarray}
	C(K)^{\scriptsize\mbox{no delay}} & = & O (h C_{k_{h -1}}^{\scriptsize\mbox{no delay}})\nonumber \\
	& = & O\Bigg(\dfrac{h}{1 - \lambda_2}\bigg(\log^{3}{K}\log{n}+\dfrac{n}{\nu}\log{K}\bigg)\Bigg).
	\label{eq:coaltime_eigen}
	\end{eqnarray}
	By Lemma~\ref{lm:mixing_time}, the total coalescence time in terms of the mixing time is given by:
	\begin{eqnarray}
	C(K)^{\scriptsize\mbox{no delay}} = O\Bigg(h\tmix^G\bigg(\log^{3}{K}+\dfrac{n}{\nu}\frac{\log{K}}{\log n}\bigg)\Bigg).
	\label{eq:coaltime_mixing}
	\end{eqnarray}
	
	Eq.~\ref{eq:coaltime_eigen} and Eq.~\ref{eq:coaltime_mixing} give the one round function computation time assuming only one packet inside each queue. Now, we will find the actual function computation time including queueing delays. 
	
	\paragraph*{Incorporating queueing delay.}
		We will use the same queueing analysis as used in the proof of Theorem~\ref{thm:fixed_time} to get a bound on the probability of a packet being delayed by a node at time $t$. So, we have $\prob{\mbox{Packet } p \mbox{ is delayed at } u \mbox{ at time }t} \leq \sum_{w\geq2} \prob{Q_t^\beta(u)=w}=\prob{Q_t^\beta(u)\geq2}$ and we know, $$\prob{\mbox{Packet } p \mbox{ is delayed at time }t}= \sum_{u \in V} \prob{\mbox{Packet } p \mbox{ is delayed at } u \mbox{ at time }t\mid X_p^t=u}\prob{X_p^t=u}.$$ So, let $c_t(\beta)=\max_{u \in V} \prob{Q_t^\beta(u)\geq2}$ be the maximum delay probability at time $t$ over all nodes in set $V$ for given data rate $\beta$. Then, at stationarity we can say that the maximum delay probability converges to $c(\beta)$ where $c(\beta)=\lim_{t\rightarrow \infty}c_t(\beta)$ and is a continuous and increasing function of $\beta$ (see Claim~1 \cite{Gillani-arXiv:2017} for details). Thus, the probability of a packet not being delayed by a node in a given time slot $t$ is $ 1 - c(\beta)$. So, combining the queueing delay with the non-delayed computation time, we have the expected number of steps by which the function is computed as $$C(K)= O\Bigg(\dfrac{h\tmix^G}{1 - c(\beta)}\bigg(\log^{3}{K}+\dfrac{n}{\nu}\frac{\log{K}}{\log n}\bigg)\Bigg).$$
	
So, we have one round function computation time i.e., assuming each source node has generated only one data packet as 
$\tcompk{1} = C(K) =  O\Bigg(\dfrac{h\tmix^G}{1 - c(\beta)}\bigg(\log^{3}{K}+\dfrac{n}{\nu}\frac{\log{K}}{\log n}\bigg)\Bigg)$, where $c(\beta)\in[0,1]$ is a continuous and increasing function of $\beta$ and $\thits$ is the worst-case hitting time of simple random walk on $G$. So,
\begin{equation}
\tcompk{\ell} \leq \ell\tcompk{1} \leq \ell\Bigg(D\dfrac{h\tmix^G}{1 - c(\beta)}\bigg(\log^{3}{K}+\dfrac{n}{\nu}\frac{\log{K}}{\log n}\bigg)\Bigg).
\label{eq:clearnce_time_k_2}
\end{equation}
where $D>1$ is a constant. Recall, $\tcomp$  is the expected function computation time of $\ell$ rounds of data arrival and so, we can write 
\begin{equation}
\tcomp = \tappk + \tcompk{\ell}.
\label{eq:tau_col_k_initial_2}
\end{equation}
Now, using the results from Eq.s \eqref{eq:tau_app_2} and \eqref{eq:clearnce_time_k_2}, we have
\begin{equation}
\tcomp \leq \ell \log eK \dfrac{b}{\beta} + \ell\left(D\dfrac{h\tmix^G}{1 - c(\beta)}\left(\log^{3}{K}+\dfrac{n}{\nu}\dfrac{\log{K}}{\log n}\right)\right)
\label{eq:tcol_k_actual_2}
\end{equation}
So, we have,
\begin{equation}
\tbarfk = \lim_{\ell\rightarrow \infty}\dfrac{\tcomp}{\ell} \leq \log eK \dfrac{b}{\beta} + \left(D\dfrac{h\tmix^G}{1 - c(\beta)}\left(\log^{3}{K}+\dfrac{n}{\nu}\dfrac{\log{K}}{\log n}\right)\right)
\label{eq:tcol_k_2}
\end{equation}
where $b,D>1$ are constants and $\thits$ is the worst-case hitting time of simple random walk on $G$.
\end{proof}

\section{Conclusion and Future Work}
\label{sec:conclusion}

In this paper, we have tried to demonstrate how random walk-based methods can be used for the in-network computation of a very general class of functions: asymmetric functions whose schema is described by a binary tree. We present lower and upper bounds on the rate for the fixed scenario. Our lower bound on rate though computed for our fixed setting is a general lower bound on rate of function computation. To the best of our knowledge, this is the first lower bound on rate for this class of problem. We also present the average function computation time under Bernoulli data generation model for both fixed and flexible model. However, our results hold for other data generation models as well, we will discuss our results in context of two different data generation models, some of the questions our setting and our results raise and also possible future directions of this work in the remaining part of this section.

\paragraph{Other data generation models} First we consider a realistic data generation model which is semi-deterministic in nature. In this model, given a $\beta > 0$ each node $u \in V_s$ generates packet with sequence number $i$ ($i > 0$) at time $\frac{i}{\beta} + N_v(0,\gamma),$ where $N_v(0,\gamma)$ is a normal random variable with mean 0 and variance $\gamma$ and $\gamma > 0$ is called the {\em clock drift parameter}. We assume that $\{N_v(0,\gamma): v \in V\}$ is an independent collection of random variables. So, for this model the maximum appearance time of $\ell$ rounds is $\tau_{app}^{\ell} \leq \frac{\ell}{\beta} + \max_{u\in V_s} N_v(0,\gamma)$. So, using this in Eq.~\eqref{eq:tau_col_k_initial} and Eq.~\eqref{eq:tau_col_k_initial_2} we get the average function computation time under the given data generation model for the \fixedalgo\ as $\dfrac{1}{\beta} + h \dfrac{\thits \log eK}{1 - c(\beta)}$ and for the \flexiblealgo\ as $\dfrac{1}{\beta} + \dfrac{h\tmix^G}{1 - c(\beta)}\bigg(\log^{3}{K}+\dfrac{n}{\nu}\dfrac{\log{K}}{\log n}\bigg)$ where variables denote the usual quantities as discussed in earlier sections.

Now, consider the other data generation model which is continuous wherein each source node has data arrivals as Poisson process. This data generation model results in a Markov process defined by the queues of the nodes. Since, every Markov process has an embedded Markov chain and we already know that the Markov chain on the queues (as seen under the Bernoulli data generation model) achieves stationarity, so the results for the Bernoulli model also hold true for this model.

\paragraph{Discussion: How general is the binary tree function schema?}
Note that any function of data consists of unary, binary or $M$-ary operations. In the function computation schema, any intermediate node representing unary operation can be merged with its parent node, i.e., the unary operation can be performed by the network vertex which performs the operation for its parent node. Thus function computation schema with only $M$-ary operations is general for $M\geq 2$. For an $M>2$ there are two possibilities. One is to create a binary tree for an $M$-ary operation by dividing the $M$-ary function into a series of binary operations. Then the time required to complete the $M$-ary function is equal to completing the operations of the equivalent binary tree which can be done using the techniques of Fixed and Flexible model analysis (see Section~\ref{subsec:computation_time}).
Another way is to look at the $M$-ary function as a symmetric function of $M$ data sources. In this case, time to compute the function can be computed using the techniques available in the literature for symmetric functions; see \cite{Mosk-PODC:2006}. 

\paragraph{Future work: Comparing the Fixed scenario and the Flexible scenario}
As discussed in Sections~\ref{sec:metrics_results}, currently we are not able to determine which of the two scenarios provides more efficient function computation. Proving that \flexiblealgo\ always does better than \fixedalgo\ is one direction that is worth pursuing since otherwise, we are not able to justify the extra information \flexiblealgo\ needs to store (each node must know the entire function schema).
However, on the grounds of fault-tolerance \flexiblealgo\ justifies the extra storage. It is more robust than
\fixedalgo\ since a single failure can disable the entire computation
in the latter case if the failure occurs at a node which is tasked
with computing a subfunction. In \flexiblealgo\ on the other hand, as
long as the sources are connected to each other the computation can
always take place since any node can perform any
subfunction. Characterising the computation time performance of
\flexiblealgo\ under suitable faults is an interesting direction to extend this
work since most real-world sensor nodes tend to be failure prone.

\paragraph{Future work: A rate result}
 In terms of analysis, the decomposition of \fixedalgo\ into a set of instances of data collection using random walks problem allows us to leverage the ideas developed in~\cite{Gillani-arXiv:2017} to characterise the rate of computation under an independent Bernoulli data generation model. However, characterising the rate of \flexiblealgo\ is not as straightforward and presents an interesting challenge that we look forward to addressing in the future. 
 
\bibliographystyle{plain}
 \bibliography{research}
 
\end{document}

%% file: func_schema.tex
\subfloat[]{
\begin{tikzpicture}[scale=1.05]
 \scriptsize
 \tikzstyle{every node} = [circle,draw=black]
 \node (a) at (-2,3) [fill=blue!20] {};
 \node (b) at (-1,3) [fill=blue!20] {};
 \node (c) at (1,3) [fill=blue!20] {};
 \node (d) at (2,3) [fill=blue!20] {};
 \node (e) at (-1.5,1.5)[font=\tiny]  {$\times$};
 \node (f) at (1.5,1.5) [font=\tiny]{$\times$};
 \node (g) at (0,0)[font=\tiny] {$+$};
 
 \node at (a) [draw=none,left]{$\T(2,0)$};
 \node at (a) [draw=none,above]{$x_1$};
 
 \node at (d) [draw=none,right]{$\T(2,3)$};
 \node at (d) [draw=none,above]{$x_4$};
 
 \node at (b) [draw=none,right]{$\T(2,1)$};
 \node at (b) [draw=none,above]{$x_2$};
 
 \node at (c) [draw=none,left] {$\T(2,2)$};
 \node at (c) [draw=none,above]{$x_3$};

 \node at (-1.7,1.5)[draw=none,left] {$\T(1,0)$};
 \node at (1.7,1.5) [draw=none,right] {$\T(1,1)$};
 \node at (-0.3,0) [draw=none,left] {$\T(0,0)$};
 \draw [->,color=red] (a) -- (e) ;
 \draw [->,color=blue] (b) -- (e) ;
 \draw [->,color=violet] (c) -- (f) ;
 \draw [->,color=green] (d) -- (f) ;
 \draw [->,color=orange] (e) -- (g) ;
 \draw [->,color=magenta] (f) -- (g) ;
 \draw [->,color=green!50!red] (g) -- (0,-1);
 
 \end{tikzpicture}
}
\subfloat[]{
\begin{tikzpicture}[scale=1.05]
 \scriptsize
 \tikzstyle{every node} = [circle,draw=black]
 \node (a) at (-2,3) [fill=blue!20] {};
 \node (b) at (-1,3) [fill=blue!20] {};
  \node (c) at (-1.5,1.5)[font=\tiny]  {$\times$};
 \node (d) at (1.5,1.5) [fill=blue!20] {};
 \node (e) at (0,0)[font=\tiny] {$+$};
 \node (f) at (3,0)[fill=blue!20] {};
\node (g) at (1.5,-1.5)[font=\tiny] {$\times$};

 \node at (a) [draw=none,left]{$\T(3,0)$};
 \node at (a) [draw=none,above]{$y_1$};
  \node at (b) [draw=none,right]{$\T(3,1)$};
 \node at (b) [draw=none,above]{$y_2$};
  \node at (-1.7,1.5)[draw=none,left] {$\T(2,0)$};
 \node at (1.7,1.5) [draw=none,right] {$\T(2,1)$};
  \node at (-0.3,0) [draw=none,left] {$\T(1,0)$};
    \node at (d) [draw=none,above]{$y_3$};
  \node at (1.75,-1.5) [draw=none,right]{$\T(0,0)$};
 \node at (f) [draw=none,right] {$\T(1,1)$};
 \node at (f) [draw=none,above]{$y_4$};

 \draw [->,color=red] (a) -- (c) ;
 \draw [->,color=blue] (b) -- (c) ;
  \draw [->,color=orange] (c) -- (e) ;
 \draw [->,color=magenta] (d) -- (e) ;
 \draw [->,color=violet] (e) -- (g) ;
 \draw [->,color=green] (f) -- (g) ;
 \draw [->,color=green!50!red] (g) -- (1.5,-2.5);
 
 \end{tikzpicture}
}

%% file: example_fixed.tex
\subfloat[]{
\begin{tikzpicture}[>=latex]
  \scriptsize
  \tikzstyle{every node} = [circle,draw=black]
  \foreach \x in {0,...,3}
  \foreach \y in {0,...,3} 
  	\node  (\x\y) at (1.5*\x,1.5*\y) {};
  
  \foreach \x in {0,...,3}
  \foreach \y [count=\yi] in {0,...,2}  
  \draw (\x\y)--(\x\yi) (\y\x)--(\yi\x) ;
  
  \draw (13) -- (02);
  \draw (22) -- (11);
  \draw (31) -- (20);
  \draw (11) -- (00);
  
  \node at (03) [draw=none,above] {$x_1$};
  \node at (03) [fill=blue!20] {};
  \node at (03) [draw=none,below right] {$\T(2,0)$};
  
  \node at (23) [draw=none,above] {$x_2$};
  \node at (23) [fill=blue!20] {};
  \node at (23) [draw=none,below left] {$\T(2,1)$};
  
  \node at (10) [draw=none,below] {$x_3$};
  \node at (10) [fill=blue!20] {};
  \node at (10) [draw=none,above right] {$\T(2,2)$};
  
  \node at (32) [draw=none,right] {$x_4$};
  \node at (32) [fill=blue!20] {};
  \node at (32) [draw=none,below left] {$\T(2,3)$};
  
  \node at (30) [draw=none,below] {$\sink$};
  \node at (30) [fill=red!20] {};
  
  \node at (22) [draw=none,above right] {$\T(1,0)$};
  \node at (22) [fill=yellow!60] {};
  \node at (22) [draw=none,below right] {$t$};
  
  \node at (11) [draw=none,above left] {$\T(1,1)$};
  \node at (11) [fill=yellow!60] {};
  \node at (11) [draw=none,below right] {$v$};
  
  \node at (31) [draw=none,above right] {$\T(0,0)$};
  \node at (31) [fill=yellow!60] {};
  \node at (31) [draw=none,below right] {$w$};

    \end{tikzpicture}
} 
\subfloat[]{
	\begin{tikzpicture}[>=latex]
	\scriptsize
	\tikzstyle{every node} = [circle,draw=black]
	\foreach \x in {0,...,3}
	\foreach \y in {0,...,3} 
	\node  (\x\y) at (1.5*\x,1.5*\y) {};
	
	\foreach \x in {0,...,3}
	\foreach \y [count=\yi] in {0,...,2}  
	\draw (\x\y)--(\x\yi) (\y\x)--(\yi\x) ;
	
	\draw (13) -- (02);
	\draw (22) -- (11);
	\draw (31) -- (20);
	\draw (11) -- (00);
	
	\node at (12) [draw=none,below right] {$u$};
	
	\node at (03) [draw=none,above] {$x_1$};
	\node at (03) [fill=blue!20] {};
	\node at (03) [draw=none,below right] {$\T(2,0)$};
    
	\node at (23) [draw=none,above] {$x_2$};
	\node at (23) [fill=blue!20] {};
	\node at (23) [draw=none,below left] {$\T(2,1)$};
    
	\node at (10) [draw=none,below] {$x_3$};
	\node at (10) [fill=blue!20] {};
	\node at (10) [draw=none,above right] {$\T(2,2)$};
    
	\node at (32) [draw=none,right] {$x_4$};
	\node at (32) [fill=blue!20] {};
	\node at (32) [draw=none,below left] {$\T(2,3)$};
    
	\node at (30) [draw=none,below] {$\sink$};
	\node at (30) [fill=red!20] {};
	
	\node at (22) [draw=none,above right] {$\T(1,0)$};
	\node at (22) [fill=yellow!60] {};
	\node at (22) [draw=none,below right] {$t$};
	
	\node at (11) [draw=none,above left] {$\T(1,1)$};
	\node at (11) [fill=yellow!60] {};
	\node at (11) [draw=none,below right] {$v$};
	
	\node at (31) [draw=none,above right] {$\T(0,0)$};
	\node at (31) [fill=yellow!60] {};
	\node at (31) [draw=none,below right] {$w$};
	
	\draw [->,very thick,color=red] (03) -- (02);
	\draw [->,very thick,color=red] (02) -- (12);
	\draw [->,very thick,color=blue] (23) -- (13);
	\draw [->,very thick,color=blue] (13) -- (12);
	
	\draw [->,very thick,color=green] (32) -- (22);
	\draw [->,very thick,color=green] (22) -- (11);
	\draw [->,very thick,color=violet] (10) -- (00);
	\draw [->,very thick,color=violet] (00) -- (11);

	\end{tikzpicture}
}
\subfloat[]{
	\begin{tikzpicture}[>=latex]
	\scriptsize
	\tikzstyle{every node} = [circle,draw=black]
	\foreach \x in {0,...,3}
	\foreach \y in {0,...,3} 
	\node  (\x\y) at (1.5*\x,1.5*\y) {};
	
	\foreach \x in {0,...,3}
	\foreach \y [count=\yi] in {0,...,2}  
	\draw (\x\y)--(\x\yi) (\y\x)--(\yi\x) ;
	
	\draw (13) -- (02);
	\draw (22) -- (11);
	\draw (31) -- (20);
	\draw (11) -- (00);
	
	\node at (03) [draw=none,above] {$x_1$};
	\node at (03) [fill=blue!20] {};
	\node at (03) [draw=none,below right] {$\T(2,0)$};
    
	\node at (23) [draw=none,above] {$x_2$};
	\node at (23) [fill=blue!20] {};
	\node at (23) [draw=none,below left] {$\T(2,1)$};
    
	\node at (10) [draw=none,below] {$x_3$};
	\node at (10) [fill=blue!20] {};
	\node at (10) [draw=none,above right] {$\T(2,2)$};
    
	\node at (32) [draw=none,right] {$x_4$};
	\node at (32) [fill=blue!20] {};
	\node at (32) [draw=none,below left] {$\T(2,3)$};
    	
	\node at (30) [draw=none,below] {$\sink$};
	\node at (30) [fill=red!20] {};
	
	\node at (12) [draw=none,below right] {$u$};
	
	\node at (22) [draw=none,above right] {$\T(1,0)$};
	\node at (22) [fill=yellow!60] {};
	\node at (22) [draw=none,below right] {$t$};
	
	\node at (11) [draw=none,above left] {$\T(1,1)$};
	\node at (11) [fill=yellow!60] {};
	\node at (11) [draw=none,below right] {$v$};
	
	\node at (31) [draw=none,above right] {$\T(0,0)$};
	\node at (31) [fill=yellow!60] {};
	\node at (31) [draw=none,below right] {$w$};
	
	\draw [->,very thick,color=red] (12) -- (22);
	\draw [->,very thick,color=magenta] (11) -- (10);
	\end{tikzpicture}
}
\subfloat[]{
	\begin{tikzpicture}[>=latex]
	\scriptsize
	\tikzstyle{every node} = [circle,draw=black]
	\foreach \x in {0,...,3}
	\foreach \y in {0,...,3} 
	\node  (\x\y) at (1.5*\x,1.5*\y) {};
	
	\foreach \x in {0,...,3}
	\foreach \y [count=\yi] in {0,...,2}  
	\draw (\x\y)--(\x\yi) (\y\x)--(\yi\x) ;
	
	\draw (13) -- (02);
	\draw (22) -- (11);
	\draw (31) -- (20);
	\draw (11) -- (00);
	
	\node at (03) [draw=none,above] {$x_1$};
	\node at (03) [fill=blue!20] {};
	\node at (03) [draw=none,below right] {$\T(2,0)$};
    
	\node at (23) [draw=none,above] {$x_2$};
	\node at (23) [fill=blue!20] {};
	\node at (23) [draw=none,below left] {$\T(2,1)$};
    
	\node at (10) [draw=none,below] {$x_3$};
	\node at (10) [fill=blue!20] {};
	\node at (10) [draw=none,above right] {$\T(2,2)$};
    
	\node at (32) [draw=none,right] {$x_4$};
	\node at (32) [fill=blue!20] {};
	\node at (32) [draw=none,below left] {$\T(2,3)$};
	
	\node at (30) [draw=none,below] {$\sink$};
	\node at (30) [fill=red!20] {};
	
	\node at (12) [draw=none,below right] {$u$};
	
	\node at (22) [draw=none,above right] {$\T(1,0)$};
	\node at (22) [fill=yellow!60] {};
	\node at (22) [draw=none,below right] {$t$};
	
	\node at (11) [draw=none,above left] {$\T(1,1)$};
	\node at (11) [fill=yellow!60] {};
	\node at (11) [draw=none,below right] {$v$};
	
	\node at (31) [draw=none,above right] {$\T(0,0)$};
	\node at (31) [fill=yellow!60] {};
	\node at (31) [draw=none,below right] {$w$};
	
	\draw [->,very thick,color=blue] (12) -- (11);
	\draw [->,very thick,color=blue] (11) -- (22);
	
	\draw [->,very thick,color=magenta] (10) -- (20);
	\draw [->,very thick,color=magenta] (20) -- (31);
	\end{tikzpicture}
}
\subfloat[]{
	\begin{tikzpicture}[>=latex]
	\scriptsize
	\tikzstyle{every node} = [circle,draw=black]
	\foreach \x in {0,...,3}
	\foreach \y in {0,...,3} 
	\node  (\x\y) at (1.5*\x,1.5*\y) {};
	
	\foreach \x in {0,...,3}
	\foreach \y [count=\yi] in {0,...,2}  
	\draw (\x\y)--(\x\yi) (\y\x)--(\yi\x) ;
	
	\draw (13) -- (02);
	\draw (22) -- (11);
	\draw (31) -- (20);
	\draw (11) -- (00);
	
	\node at (03) [draw=none,above] {$x_1$};
	\node at (03) [fill=blue!20] {};
	\node at (03) [draw=none,below right] {$\T(2,0)$};
    
	\node at (23) [draw=none,above] {$x_2$};
	\node at (23) [fill=blue!20] {};
	\node at (23) [draw=none,below left] {$\T(2,1)$};
    
	\node at (10) [draw=none,below] {$x_3$};
	\node at (10) [fill=blue!20] {};
	\node at (10) [draw=none,above right] {$\T(2,2)$};
    
	\node at (32) [draw=none,right] {$x_4$};
	\node at (32) [fill=blue!20] {};
	\node at (32) [draw=none,below left] {$\T(2,3)$};
	
	\node at (30) [draw=none,below] {$\sink$};
	\node at (30) [fill=red!20] {};
	
	\node at (12) [draw=none,below right] {$u$};
	
	\node at (22) [draw=none,above right] {$\T(1,0)$};
	\node at (22) [fill=yellow!60] {};
	\node at (22) [draw=none,below right] {$t$};
	
	\node at (11) [draw=none,above left] {$\T(1,1)$};
	\node at (11) [fill=yellow!60] {};
	\node at (11) [draw=none,below right] {$v$};
	
	\node at (31) [draw=none,above right] {$\T(0,0)$};
	\node at (31) [fill=yellow!60] {};
	\node at (31) [draw=none,below right] {$w$};
	
	\draw [->,very thick,color=orange] (22)--(21);
	\draw [->,very thick,color=orange] (21) -- (31);
	\draw [->,very thick,color=green!50!red] (31) -- (30);
	\end{tikzpicture}
}

%% file: network_example.tex
\subfloat[]{
\begin{tikzpicture}[>=latex]
  \scriptsize
  \tikzstyle{every node} = [circle,draw=black]
  \foreach \x in {0,...,3}
  \foreach \y in {0,...,3} 
  	\node  (\x\y) at (1.5*\x,1.5*\y) {};
  
  \foreach \x in {0,...,3}
  \foreach \y [count=\yi] in {0,...,2}  
  \draw (\x\y)--(\x\yi) (\y\x)--(\yi\x) ;
  
  \draw (13) -- (02);
  \draw (22) -- (11);
  \draw (31) -- (20);
  \draw (11) -- (00);
  
  \node at (03) [draw=none,above] {$x_1$};
  \node at (03) [fill=blue!20] {};
  \node at (03) [draw=none,below left] {$\T(2,0)$};
  
  \node at (23) [draw=none,above] {$x_2$};
  \node at (23) [fill=blue!20] {};
  \node at (23) [draw=none,below left] {$\T(2,1)$};
  
  \node at (10) [draw=none,below] {$x_3$};
  \node at (10) [fill=blue!20] {};
  \node at (10) [draw=none,above left] {$\T(2,2)$};
  
  \node at (32) [draw=none,right] {$x_4$};
  \node at (32) [fill=blue!20] {};
  \node at (32) [draw=none,below left] {$\T(2,3)$};
  
  \node at (30) [draw=none,below] {$\sink$};
  \node at (30) [fill=red!20] {};
  
    \end{tikzpicture}
} 
\subfloat[]{
	\begin{tikzpicture}[>=latex]
	\scriptsize
	\tikzstyle{every node} = [circle,draw=black]
	\foreach \x in {0,...,3}
	\foreach \y in {0,...,3} 
	\node  (\x\y) at (1.5*\x,1.5*\y) {};
	
	\foreach \x in {0,...,3}
	\foreach \y [count=\yi] in {0,...,2}  
	\draw (\x\y)--(\x\yi) (\y\x)--(\yi\x) ;
	
	\draw (13) -- (02);
	\draw (22) -- (11);
	\draw (31) -- (20);
	\draw (11) -- (00);
	
	\node at (03) [draw=none,above] {$x_1$};
	\node at (03) [fill=blue!20] {};
	
	\node at (23) [draw=none,above] {$x_2$};
	\node at (23) [fill=blue!20] {};
	
	\node at (10) [draw=none,below] {$x_3$};
	\node at (10) [fill=blue!20] {};
	
	\node at (32) [draw=none,right] {$x_4$};
	\node at (32) [fill=blue!20] {};
	
	\draw [->,very thick,color=red] (03) -- (02);
	\draw [->,very thick,color=red] (02) -- (12);
	\draw [->,very thick,color=blue] (23) -- (13);
	\draw [->,very thick,color=blue] (13) -- (12);
	\node at (12) [fill=yellow!60] {};
	\node at (12) [draw=none,above right] {$\T(1,0)$};
	
	\draw [->,very thick,color=green] (32) -- (22);
	\draw [->,very thick,color=green] (22) -- (11);
	\draw [->,very thick,color=violet] (10) -- (00);
	\draw [->,very thick,color=violet] (00) -- (11);
	\node at (11) [fill=yellow!60] {};
	\node at (11) [draw=none,above left] {$\T(1,1)$};
	
	\node at (30) [draw=none,below] {$\sink$};
	\node at (30) [fill=red!20] {};
	
	\node at (12) [draw=none,below right] {$u$};
	\node at (11) [draw=none,below right] {$v$};
	\node at (31) [draw=none,below right] {$w$};
	
	\end{tikzpicture}
}
\subfloat[]{
	\begin{tikzpicture}[>=latex]
	\scriptsize
	\tikzstyle{every node} = [circle,draw=black]
	\foreach \x in {0,...,3}
	\foreach \y in {0,...,3} 
	\node  (\x\y) at (1.5*\x,1.5*\y) {};
	
	\foreach \x in {0,...,3}
	\foreach \y [count=\yi] in {0,...,2}  
	\draw (\x\y)--(\x\yi) (\y\x)--(\yi\x) ;
	
	\draw (13) -- (02);
	\draw (22) -- (11);
	\draw (31) -- (20);
	\draw (11) -- (00);
	
	\node at (03) [draw=none,above] {$x_1$};
	\node at (03) [fill=blue!20] {};
	
	\node at (23) [draw=none,above] {$x_2$};
	\node at (23) [fill=blue!20] {};
	
	\node at (10) [draw=none,below] {$x_3$};
	\node at (10) [fill=blue!20] {};
	
	\node at (32) [draw=none,right] {$x_4$};
	\node at (32) [fill=blue!20] {};
	
	\draw [->,very thick,color=orange] (12)--(22);
	\draw [->,very thick,color=orange] (22)--(21);
	\draw [->,very thick,color=orange] (21)--(31);
	
	\draw [->,very thick,color=magenta] (11) -- (10);
	\draw [->,very thick,color=magenta] (10) -- (20);
	\draw [->,very thick,color=magenta] (20) -- (31);
	\node at (31) [fill=yellow!60] {};
	\node at (31) [draw=none,above right] {$\T(0,0)$};
	
	\draw [->,very thick,color=green!50!red] (31) -- (30);
	
	\node at (30) [draw=none,below] {$\sink$};
	\node at (30) [fill=red!20] {};
	
	\node at (12) [draw=none,below right] {$u$};
	\node at (11) [draw=none,below right] {$v$};
	\node at (31) [draw=none,below right] {$w$};
	\end{tikzpicture}
}